\documentclass[journal,letterpaper,10pt,twocolumn]{IEEEtran}

\input{definition}
\usepackage{etoolbox}
\usepackage[export]{adjustbox}
\usepackage{tocbibind}
\usepackage[hidelinks]{hyperref}

\newcommand\simpath{Simulation}
\newcommand\obs{\mathrm{obs}}
\newcommand\jacob[1]{\nabla_{\!#1}}
\newcommand\rbf{\mathrm{rbf}}

\title{Adaptive learning in two-player Stackelberg games with application to network security}
\author{Guosong~Yang, Radha~Poovendran, and Jo\~{a}o~P.~Hespanha%
    \thanks{G.~Yang and J.~P.~Hespanha are with the Center for Control, Dynamical Systems, and Computation, University of California, Santa Barbara, CA 93106, USA. Email: \texttt{\{guosongyang,~hespanha\}@ucsb.edu}.}
    \thanks{R.~Poovendran is with the Department of Electrical and Computer Engineering, University of Washington, Seattle, WA 98195, USA. Email: \texttt{rp3@uw.edu}.}}

\begin{document}
\maketitle

\begin{abstract}
We study a two-player Stackelberg game with incomplete information such that the follower's strategy belongs to a known family of parameterized functions with an unknown parameter vector. We design an adaptive learning approach to simultaneously estimate the unknown parameter and minimize the leader's cost, based on adaptive control techniques and hysteresis switching. Our approach guarantees that the leader's cost predicted using the parameter estimate becomes indistinguishable from its actual cost in finite time, up to a preselected, arbitrarily small error threshold. Also, the first-order necessary condition for optimality holds asymptotically for the predicted cost. Additionally, if a persistent excitation condition holds, then the parameter estimation error becomes bounded by a preselected, arbitrarily small threshold in finite time as well. For the case where there is a mismatch between the follower's strategy and the parameterized function that is known to the leader, our approach is able to guarantee the same convergence results for error thresholds larger than the size of the mismatch. The algorithms and the convergence results are illustrated via a simulation example in the domain of network security.
\end{abstract}

\section{Introduction}\label{sec:intro}
A modern engineering system often involves multiple self-interested decision makers whose actions have mutual consequences. Examples of such systems include communications over a shared network with limited capacity, and computer programs sharing limited computational resources. Game theory provides a systematic framework for modeling the cooperation and conflict between these so-called strategic players \cite{FudenbergTirole1991}, and has been widely applied to areas such as robust design, resource allocation, and network security \cite{BasarOlsder1999,AlpcanBasar2010,Hespanha2017}.

In game theory, a fundamental question is whether the players can converge to a Nash equilibrium---a tuple of strategies for which no one has a unilateral incentive to change---if they play the game iteratively and adjust their strategies based on historical outcomes. A primary example of such a learning process is fictitious play \cite{Brown1951,Robinson1951}, in which every player believes that the opponents are playing constant mixed strategies in agreement with the empirical distributions of their past actions, and plays the corresponding best response. Another well-known example is gradient response \cite{BrownNeumann1952,Rosen1965}, in which each player adjusts its strategy according to the corresponding gradient of its cost function. Fictitious play and gradient response have attracted significant research interests \cite{FudenbergLevine1998,Hart2005} and have been used in applications such as multiagent reinforcement learning \cite{BusoniuBabuskaSchutter2008} and distributed control \cite{MardenShamma2015}.

In this paper, we propose an adaptive learning approach for a hierarchical game model introduced by Stackelberg \cite{Stackelberg2011}. In a two-player Stackelberg game, one player (called the \emph{leader}) selects its action first, and then the other player (called the \emph{follower}), informed of the leader's choice, selects its own action. Therefore, a follower's strategy in a Stackelberg game can be viewed as a function that specifies an action in response to each leader's possible action.

Stackelberg games provide a natural framework for understanding systems with asymmetric information, a common feature of many network problems such as routing \cite{KorilisLazarOrda1997}, scheduling \cite{Roughgarden2004}, and channel allocation \cite{BloemAlpcanBasar2007}. They are especially useful for modeling security problems, where the defender (leader) is usually unaware of the attacker's objective and capabilities a priori, whereas the attacker (follower) is able to observe the defender's strategy and attack after careful planning. A class of Stackelberg games called Stackelberg Security Games have been applied to various real-world security domains and have lead to practical implementations such as the ARMOR program at the Los Angeles International Airport \cite{PitaJainMareckiOrdonezPortwayTambeWesternParuchuriKraus2008}, the IRIS program used by the US Federal Air Marshals \cite{TsaiRathiKiekintveldOrdonezTambe2009}, and counterterrorism programs for crucial infrastructures such as power grid and oil reserves \cite{BrownCarlyleSalmeronWood2005,BrownCarlyleSalmeronWood2006}.

Asymmetric information often leads to scenarios with no Nash equilibrium but a Stackelberg equilibrium, as sufficient conditions for the existence of the former are much stronger than those for the latter (see, e.g., \cite[p.~181]{BasarOlsder1999}). For these scenarios, learning options like fictitious play and gradient response cannot be applied, and novel approaches are needed to achieve convergence to a Stackelberg equilibrium. Existing results on learning in Stackelberg games are limited to linear or quadratic costs and finite action sets \cite{BrucknerScheffer2011,MareckiTesauroSegal2012,BlumHaghtalabProcaccia2014}, which are too restrictive for many applications including network security.

In this paper, we study a Stackelberg game between two players with continuous action sets. We consider the scenario where the leader only has partial knowledge of the follower's action set and cost function. As a result, the follower's strategy belongs to a family of parameterized functions that is known to the leader, but the actual value of the parameter vector is unknown. Our main contribution is an adaptive learning approach described in Section~\ref{sec:alg}, which simultaneously estimates the unknown parameter based on the follower's past actions and minimizes the leader's cost predicted using the parameter estimate. The approach is designed based on adaptive control techniques, and utilizes projections and hysteresis switching to ensure feasibility of solutions and finite-time convergence of the parameter estimate.

In Section~\ref{sec:conv}, we prove that the leader's predicted cost is guaranteed to become indistinguishable from its actual cost in finite time, up to a preselected, arbitrarily small error threshold. Also, the first-order necessary condition for optimality holds asymptotically for the predicted cost. Additionally, if a persistent excitation condition holds, then the parameter estimation error is guaranteed to become bounded by a preselected, arbitrarily small threshold in finite time as well. Our proof provides a rigorous treatment for the existence and convergence of solutions to the discontinuous dynamics that results from projections and switching, based on tools from differential inclusions theory. In particular, we establish an invariance principle for projected gradient descent in continuous time, which is of independent interest and is novel to the best of our knowledge.

In Section~\ref{sec:mismatch}, we extend the adaptive learning approach to the case where the parameterized function that is known to the leader does not match the follower's strategy perfectly. We prove that that our approach can be adjusted to guarantee the same convergence results for preselected error thresholds that are larger than the size of the mismatch.

In Section~\ref{sec:sim}, the algorithms and the convergence results are illustrated via a simulation example motivated by link-flooding distributed denial-of-service (DDoS) attacks, such as the Crossfire attack \cite{KangLeeGligor2013}. Section~\ref{sec:end} concludes the paper with a brief summary and an outlook on future research topics.

A preliminary version for some of these results appeared in the conference paper \cite{YangPoovendranHespanha2019}. The current paper improves \cite{YangPoovendranHespanha2019} by adding a notion of ``practical'' Stackelberg equilibrium, removing unnecessary assumptions, substantiating the results with complete proofs and clarifying remarks, and providing a more realistic and elaborated simulation example.

\emph{Notations:} Let $ \R_+ := [0, \infty) $ and $ \N := \{0, 1, \ldots\} $. Denote by $ I_n $ the identity matrix in $ \R^{n \times n} $, or simply $ I $ when the dimension is implicit. Denote by $ \|\cdot\| $ the Euclidean norm for vectors and the (induced) Euclidean norm for matrices. For two vectors $ v_1 $ and $ v_2 $, denote by $ (v_1, v_2) := (v_1^\top, v_2^\top)^\top $ their concatenation. For a set $ \cS \subset \R^n $, denote by $ \partial\cS $ and $ \overline\cS $ its boundary and closure, respectively. A signal $ r: \R_+ \to \R^n $ is of class $ \Linf $ if $ \sup_{t \geq 0} \|r(t)\| $ is finite. Denote by $ s \bB(x) $ the closed ball of radius $ s \geq 0 $ centered at $ x \in \R^n $, that is, $ s \bB(x) := \{z \in \R^n: \|z - x\| \leq s\} $.

\section{Problem formulation}
Consider a two-player game $ (\cR, \cA, J, H) $, where $ \cR \subset \R^{n_r} $ and $ \cA \subset \R^{n_a} $ are the \emph{action sets} of the first and the second player, respectively, and $ J: \cR \times \cA \to \R $ and $ H: \cA \times \cR \to \R $ are the corresponding \emph{cost functions} to be minimized. We are interested in a hierarchical game model proposed by Stackelberg \cite{Stackelberg2011}, where the first player (called the \emph{leader}) selects its action $ r \in \cR $ first, and then the second player (called the \emph{follower}), informed of the leader's choice, selects its own action $ a \in \cA $. Formally, a Stackelberg equilibrium is defined as follows; see, e.g., \cite[Sec.~3.1]{FudenbergTirole1991} or \cite[Def.~4.6 and~4.7, p.~179 and~180]{BasarOlsder1999}.
\begin{dfn}[Stackelberg equilibrium]\label{dfn:stbg}
Given a game defined by $ (\cR, \cA, J, H) $, an action $ r^* \in \cR $ is called a \emph{Stackelberg equilibrium action} for the leader if
\begin{equation}\label{eq:stbg}
    J^* := \inf_{r \in \cR} \max_{a \in \beta_a(r)} J(r, a) = \max_{a \in \beta_a(r^*)} J(r^*, a),
\end{equation}
where\footnote{In this paper, we only consider compact action sets and continuous cost functions; thus for each $ r \in \cR $, the set $ \beta_a(r) $ is nonempty and compact. However, the map $ r \mapsto \max_{a \in \beta_a(r)} J(r, a) $ may still be discontinuous.}
\begin{equation}\label{eq:opt-res-set}
    \beta_a(r) := \argmin_{a \in \cA} H(a, r)
\end{equation}
denotes the set of follower's best responses against a leader's action $ r \in \cR $, and $ J^* $ is known as the \emph{Stackelberg cost} for the leader; for an $ \varepsilon > 0 $, a routing action $ r^*_\varepsilon \in \cR $ is called an \emph{$ \varepsilon $ Stackelberg action} for the leader if \eqref{eq:stbg} is replaced by
\begin{equation*}
    \max_{a \in \beta_a(r^*_\varepsilon)} J(r^*_\varepsilon, a) \leq J^* + \varepsilon.
\end{equation*}
\end{dfn}

The Stackelberg cost provides a cost that the leader is able to guarantee against a rational follower. Depending on the follower's actual strategy, the leader may be able to achieve a better (smaller) cost while playing a Stackelberg equilibrium action. In practice, it is possible that no Stackelberg equilibrium action exists but the leader is able to achieve essentially the Stackelberg cost by playing an $ \varepsilon $ Stackelberg action while selecting a sufficiently small $ \varepsilon > 0 $; see also the discussion after Assumption~\ref{ass:reg} and the example in Section~\ref{sec:sim}.

We consider games with perfect but incomplete information, where the leader only has partial knowledge of the follower's action set and cost function and cannot predict the follower's action accurately. Specifically, the follower's strategy is an unknown function $ f: \cR \to \cA $ such that $ a = f(r) \in \beta_a(r) $ for all $ r \in \cR $. However, $ f $ belongs to a family of parameterized functions $ \{r \mapsto \hat f(\hat\theta, r): \hat\theta \in \Theta\} $ with a parameter set $ \Theta \subset \R^{n_\theta} $, that is, there is an \emph{actual value} $ \theta \in \Theta $ such that
\begin{equation}\label{eq:match}
    a = f(r) = \hat f(\theta, r) \qquad \forall r \in \cR.
\end{equation}
The parameterized function $ \hat f: \Theta \times \cR \to \R^{n_a} $ (including the parameter set $ \Theta $) is known to the leader, but the actual value $ \theta $ is unknown. The follower's action set $ \cA $ is also unknown except that $ \cA \subset \{\hat f(\hat\theta, r): \hat\theta \in \Theta,\, r \in \cR\} $ as implied by \eqref{eq:match}.

In practice, assuming that the follower's strategy belongs to a known family of parameterized functions introduces little loss of generality, as the follower's strategy can always be approximated on a compact set up to an arbitrary precision as a finite weighted sum of a preselected class of basis functions. An example of such an approximation is the \emph{radial basis function (RBF)} model \cite{FangLiSudjianto2005}, in which the leader assumes
\begin{equation*}
    \hat f(\theta, r) = \sum_{j=1}^{n_\theta} \theta_j F_j(r) = \sum_{j=1}^{n_\theta} \theta_j \phi(\|r - r_j^c\|),
\end{equation*}
where $ \phi: \R_+ \to \R^{n_a} $ is an RBF and each $ F_j: \cR \to \R^{n_a} $ is a kernel centered at $ r_j^c $. Note that in the RBF model, the approximation is affine with respect to the unknown parameter, which is also common in many other widely used approximation models such as \emph{orthogonal polynomials} and \emph{multivariate splines} \cite{FangLiSudjianto2005}. This motivates restricting our attention to affine maps $ \hat\theta \mapsto \hat f(\hat\theta, r) $. The following assumption captures this and other generic regularity conditions that we use to guarantee the existence of an $ \varepsilon $ Stackelberg action and the convergence of our parameter estimation algorithms.
\begin{ass}[Regularity]\label{ass:reg}
The follower's action set $ \cA $ is compact, and the leader's action set $ \cR $ and the parameter set $ \Theta $ are convex and compact; the follower's cost function $ H $ is continuous, the leader's cost function $ J $ and the parameterized function $ \hat f $ are continuously differentiable, and the map $ \hat\theta \mapsto \hat f(\hat\theta, r) $ is affine for each fixed $ r \in \cR $.
\end{ass}

Under Assumption~\ref{ass:reg}, there exists an $ \varepsilon $ Stackelberg action for each $ \varepsilon > 0 $ \cite[Prop.~4.2, p.~180]{BasarOlsder1999}. These conditions are much weaker than the standard sufficient conditions for the existence of a Stackelberg equilibrium action \cite[Th.~4.8, p.~180]{BasarOlsder1999}, which are in turn much weaker than those for a Nash equilibrium \cite[p.~181]{BasarOlsder1999}. Therefore, they are consistent with our interest in games with no Nash equilibrium but a ``practical'' Stackelberg equilibrium; see also the example in Section~\ref{sec:sim}.

We denote by $ \jacob{r} J(r, a) $ and $ \jacob{a} J(r, a) $ the gradients of the maps $ r \mapsto J(r, a) $ and $ a \mapsto J(r, a) $, respectively, and by $ \jacob{\theta} \hat f(r) $ and $ \jacob{r} \hat f(\hat\theta, r) $ the Jacobian matrices of the maps $ \hat\theta \mapsto \hat f(\hat\theta, r) $ and $ r \mapsto \hat f(\hat\theta, r) $, respectively.\footnote{To be consistent with the definition of Jacobian matrix, we take gradients as row vectors.} In particular, the Jacobian matrix $ \jacob{\theta} \hat f(r) $ is independent of $ \hat\theta $ due to the affine condition in Assumption~\ref{ass:reg}.

Our goal is to adjust the leader's action $ r $ to minimize its cost $ J(r, a) $ for the follower's action $ a = f(r) = \hat f(\theta, r) $, that is, to solve the optimization problem\footnote{Clearly, as the follower's response $ \hat f(\theta, r) \in \beta_a(r) $ for all $ r \in \cR $, the leader's optimal cost is upper bounded by its Stackelberg cost, i.e., $ \min_{r \in \cR} J(r, \hat f(\theta, r)) \leq J^* $.}
\begin{equation}\label{eq:op-act-cost}
    \min_{r \in \cR} J \big( r, \hat f(\theta, r) \big),
\end{equation}
based on past observations of the follower's action $ a = \hat f(\theta, r) $ and the leader's cost $ J(r, a) $, but without knowing the actual value $ \theta $. Our approach to solve this problem combines the following two components:
\begin{enumerate}
    \item Construct a \emph{parameter estimate} $ \hat\theta $ that approaches the actual value $ \theta $.
    \item Adjust the leader's action $ r $ based on a gradient descent method to minimize its \emph{predicted cost}
        \begin{equation*}
            \hat J(r, \hat\theta) := J \big( r, \hat f(\hat\theta, r) \big),
        \end{equation*}
        that is, to solve the optimization problem
        \begin{equation}\label{eq:op-est-cost}
            \min_{r \in \cR} \hat J(r, \hat\theta) = \min_{r \in \cR} J \big( r, \hat f(\hat\theta, r) \big).
        \end{equation}
\end{enumerate}
In this paper, our design and analysis are formulated using continuous-time dynamics, which is common in the literature of learning in game theory \cite{FudenbergLevine1998,Hart2005}.

\section{Estimation and optimization}\label{sec:alg}
To specify the adaptive algorithms for estimating the actual value $ \theta $ and optimizing the leader's action $ r $, we recall the following notions and basic properties from convex analysis; for more details, see, e.g., \cite[Ch.~6]{RockafellarWets1998} or \cite[Sec.~5.1]{Aubin1991}.

For a closed convex set $ \cC \subset \R^n $ and a point $ v \in \R^n $, we denote by $ [v]_\cC $ the \emph{projection} of $ v $ onto $ \cC $, that is,
\begin{equation*}
    [v]_\cC := \argmin_{w \in \cC} \|w - v\|.
\end{equation*}
The projection $ [v]_\cC $ exists and is unique as the set $ \cC $ is closed and convex, and satisfies $ [v]_\cC = v $ if $ v \in \cC $.

For a convex set $ \cS \subset \R^n $ and a point $ x \in \cS $, we denote by $ T_\cS(x) $ the \emph{tangent cone} to $ \cS $ at $ x $, that is,
\begin{equation}\label{eq:tan-cone}
    T_\cS(x) := \overline{\{h (z - x): z \in \cS,\, h > 0\}},
\end{equation}
and by $ N_\cS(x) $ the \emph{normal cone} to $ \cS $ at $ x $, that is,
\begin{equation}\label{eq:normal-cone}
    N_\cS(x) := \{v \in \R^n: v^\top w \leq 0 \text{ for all } w \in T_\cS(x)\}.
\end{equation}
The sets $ T_\cS(x) $ and $ N_\cS(x) $ are closed and convex, and satisfy $ T_\cS(x) = \R^n $ and $ N_\cS(x) = \{0\} $ if $ x \in \cS\backslash\partial\cS $. Moreover, we have
\begin{equation}\label{eq:tan-proj-cone}
    [v]_{T_\cS(x)} \in T_\cS(x), \quad v - [v]_{T_\cS(x)} \in N_\cS(x)
\end{equation}
and
\begin{equation}\label{eq:tan-proj-zero}
    \big( v - [v]_{T_\cS(x)} \big)^\top [v]_{T_\cS(x)} = 0
\end{equation}
for all $ v \in \R^n $ and $ x \in \cS  $.

\subsection{Parameter estimation}\label{ssec:alg-est}
We construct the parameter estimate $ \hat\theta $ by comparing past observations of the follower's action $ a = \hat f(\theta, r) $ and the leader's cost $ J(r, a) $ with the corresponding predicted values $ \hat f(\hat\theta, r) $ and $ \hat J(r, \hat\theta) = J \big( r, \hat f(\hat\theta, r) \big) $ computed using $ \hat\theta $. Their difference is defined as the \emph{observation error}
\begin{equation}\label{eq:obs-err}
    e_\obs := \begin{bmatrix}
        \hat f(\hat\theta, r) - a \\
        \hat J(r, \hat\theta) - J(r, a)
    \end{bmatrix}.
\end{equation}
We develop an estimation algorithm based on the observation error $ e_\obs $ so that the norm of the \emph{estimation error} $ \hat\theta - \theta $ is monotonically decreasing, regardless of how the leader's action $ r $ is being adjusted.

First, we establish a relation between the observation error $ e_\obs $ and the estimation error $ \hat\theta - \theta $.
\begin{lem}\label{lem:obs-est}
The observation error $ e_\obs $ satisfies
\begin{equation}\label{eq:obs-est}
    e_\obs = K(r, a, \hat\theta) (\hat\theta - \theta)
\end{equation}
with the gain matrix
\begin{equation}\label{eq:est-gain}
    K(r, a, \hat\theta) := \begin{bmatrix}
        I \\
        \int_{0}^{1} \jacob{a} J \big( r, \rho \hat f(\hat\theta, r) + (1 - \rho)\, a \big) \d\rho
    \end{bmatrix} \jacob{\theta} \hat f(r).
\end{equation}
\end{lem}
\begin{proof}
See Appendix~\ref{apx:obs-est}.
\end{proof}

Following Lemma~\ref{lem:obs-est}, the observation error $ e_\obs $ would be zero if the current estimate $ \hat\theta $ of the actual value $ \theta $ was correct. However, in most interesting scenarios, the dimension $ n_\theta $ of the parameter vector $ \theta $ is much larger than the dimension $ n_a + 1 $ of the observation error $ e_\obs $; thus the gain matrix $ K(r, a, \hat\theta) $ cannot be invertible, and a zero value for the observation error $ e_\obs $ does not imply that the estimate $ \hat\theta $ is correct.

We propose the following estimation algorithm to drive the parameter estimate $ \hat\theta $ towards the actual value $ \theta $:
\begin{equation}\label{eq:est-dyn}
    \dot{\hat\theta} = \big[ {-\lambda_e} K(r, a, \hat\theta)^\top e_\obs \big]_{T_\Theta(\hat\theta)}
\end{equation}
with the gain matrix $ K(r, a, \hat\theta) $ defined by \eqref{eq:est-gain} and the \emph{switching signal} $ \lambda_e: \R_+ \to \{0,\, \lambda_\theta\} $ defined by
\begin{equation}\label{eq:est-sw}
    \lambda_e(t) := \begin{cases}
        \lambda_\theta &\text{if } \|e_\obs(t)\| \geq \varepsilon_\obs; \\
        \lim_{s \nearrow t} \lambda_e(s) &\text{if } \|e_\obs(t)\| \in (\varepsilon_\obs', \varepsilon_\obs); \\
        0 &\text{if } \|e_\obs(t)\| \leq \varepsilon_\obs',
    \end{cases}
\end{equation}
and $ \lambda_e(0) := \lambda_\theta $ if $ \|e_\obs(0)\| \in (\varepsilon_\obs', \varepsilon_\obs) $, where $ \varepsilon_\obs > \varepsilon_\obs'  > 0 $ and $ \lambda_\theta > 0 $ are preselected constants. Several comments are in order: First, the gain matrix $ K(r, a, \hat\theta) $ depends on the parameter estimate $ \hat\theta $ but not on the actual value $ \theta $, so \eqref{eq:est-dyn} can be implemented without knowing $ \theta $. Second, the projection $ [\cdot]_{T_\Theta(\hat\theta)} $ onto the tangent cone $ T_\Theta(\hat\theta) $ ensures that the parameter estimate $ \hat\theta $ remains inside the parameter set $ \Theta $ \cite{Aubin1991}; see Appendix~\ref{apx:proj-dyn} for more details. Finally, the right-continuous, piecewise constant switching signal $ \lambda_e $ is designed so that the adaptation is on when $ \|e_\obs\| \geq \varepsilon_\obs $ and off when $ \|e_\obs\| \leq \varepsilon_\obs' $, with a hysteresis switching rule that avoids chattering. The key feature of \eqref{eq:est-dyn} is that the estimation error $ \hat\theta - \theta $ satisfies
\begin{multline*}
    \frac{\d\|\hat\theta - \theta\|^2}{\d t} = 2 (\hat\theta - \theta)^\top \big[ {-\lambda_e} K(r, a, \hat\theta)^\top e_\obs \big]_{T_\Theta(\hat\theta)} \\
    \leq 2 (\hat\theta - \theta)^\top \big( {-\lambda_e} K(r, a, \hat\theta)^\top e_\obs \big) = -2 \lambda_e \|e_\obs\|^2,
\end{multline*}
where the inequality follows from \eqref{eq:tan-cone}--\eqref{eq:tan-proj-cone}. We thus conclude that the estimation algorithm \eqref{eq:est-dyn} with the switching signal \eqref{eq:est-sw} guarantees
\begin{equation}\label{eq:est-lya}
    \frac{\d\|\hat\theta - \theta\|^2}{\d t} \leq -2 \lambda_e \|e_\obs\|^2 \leq 0,
\end{equation}
which implies that $ \|\hat\theta - \theta\| $ is monotonically decreasing and will not stop approaching zero unless $ \|e_\obs\| < \varepsilon_\obs $. In the convergence analysis in Section~\ref{sec:conv}, we will show that the adaptation of the parameter estimate $ \hat\theta $ stops in finite time, and the observation error $ e_\obs $ satisfies $ \|e_\obs\| < \varepsilon_\obs $ afterward.

\subsection{Cost minimization}\label{ssec:alg-ctrl}
Several options are available to adjust the leader's action $ r $, but in this paper our analysis will focus on a gradient descent method, which is fairly robust for a wide range of problems. Our ultimate goal is to minimize the leader's cost $ J(r, a) = J \big( r, \hat f(\theta, r) \big) $. However, computing the gradient descent direction of the actual cost requires knowledge of the actual value $ \theta $. Therefore, we minimize instead the predicted cost $ \hat J(r, \hat\theta) = J \big( r, \hat f(\hat\theta, r) \big) $, which depends instead on the parameter estimate $ \hat\theta $. This change in objective is justified by the property that $ \|\hat J(r, \hat\theta) - J(r, a)\| \leq \|e_\obs\| < \varepsilon_\obs $ holds after a finite time, which will be established in Section~\ref{sec:conv}.

The time derivative of the predicted cost $ \hat J(r, \hat\theta) $ is given by
\begin{equation}\label{eq:est-cost-grad}
    \dot{\hat J}(r, \hat\theta) = \jacob{r} \hat J(r, \hat\theta)\, \dot r + \jacob{\theta} \hat J(r, \hat\theta)\, \dot{\hat\theta},
\end{equation}
with
\begin{equation*}
\begin{aligned}
    \jacob{r} \hat J(r, \hat\theta) &:= \jacob{r} J \big( r, \hat f(\hat\theta, r) \big) + \jacob{a} J \big( r, \hat f(\hat\theta, r) \big) \jacob{r} \hat f(\hat\theta, r), \\
    \jacob{\theta} \hat J(r, \hat\theta) &:= \jacob{a} J \big( r, \hat f(\hat\theta, r) \big) \jacob{\theta} \hat f(r).
\end{aligned}
\end{equation*}
Here $ \jacob{r} J \big( r, \hat f(\hat\theta, r) \big) $ denotes the gradient of the map $ r \mapsto J(r, \hat a) $ at $ \hat a = \hat f(\hat\theta, r) $; thus $ \jacob{r} \hat J(r, \hat\theta) $ and $ \jacob{\theta} \hat J(r, \hat\theta) $ are the gradients of the maps $ r \mapsto \hat J(r, \hat\theta) = J \big( r, \hat f(\hat\theta, r) \big) $ and $ \hat\theta \mapsto \hat J(r, \hat\theta) $, respectively. As we will establish that the adaptation of the parameter estimate $ \hat\theta $ stops in finite time, we neglect the term with $ \dot{\hat\theta} $ in \eqref{eq:est-cost-grad} and focus exclusively in adjusting $ r $ along the gradient descent direction of $ r \mapsto \hat J(r, \hat\theta) $. This motivates the following optimization algorithm to adjust the leader's action:
\begin{equation}\label{eq:ctrl-dyn}
    \dot r = \big[ {-\lambda_r} \jacob{r} \hat J(r, \hat\theta)^\top \big]_{T_\cR(r)},
\end{equation}
where $ \lambda_r > 0 $ is a preselected constant. The projection $ [\cdot]_{T_\cR(r)} $ onto the tangent cone $ T_\cR(r) $ ensures that the leader's action $ r $ remains inside the action set $ \cR $ \cite{Aubin1991}; see Appendix~\ref{apx:proj-dyn} for more details. From \eqref{eq:est-cost-grad} and \eqref{eq:ctrl-dyn} we conclude that
\begin{equation*}
\begin{aligned}
    \dot{\hat J}(r, \hat\theta) &= \jacob{r} \hat J(r, \hat\theta) \big[ {-\lambda_r} \jacob{r} \hat J(r, \hat\theta)^\top \big]_{T_\cR(r)} + \jacob{\theta} \hat J(r, \hat\theta)\, \dot{\hat\theta} \\
    &= -\big\| \big[ {-\lambda_r} \jacob{r} \hat J(r, \hat\theta)^\top \big]_{T_\cR(r)} \big\|^2 \big/ \lambda_r + \jacob{\theta} \hat J(r, \hat\theta)\, \dot{\hat\theta} \\
    &= -\|\dot r\|^2/\lambda_r + \jacob{\theta} \hat J(r, \hat\theta)\, \dot{\hat\theta},
\end{aligned}
\end{equation*}
where the second equality follows from \eqref{eq:tan-proj-zero}. We thus conclude that the optimization algorithm \eqref{eq:ctrl-dyn} guarantees
\begin{equation*}
    \dot{\hat\theta} = 0 \implies \dot{\hat J}(r, \hat\theta) \leq -\|\dot r\|^2/\lambda_r \leq 0.
\end{equation*}
In the convergence analysis in Section~\ref{sec:conv}, we will show that the leader's action $ r $ converges asymptotically to the set of points for which the first-order necessary condition for optimality holds for the optimization problem \eqref{eq:op-est-cost}.

\section{Convergence analysis}\label{sec:conv}
We now present the main result of this paper:
\begin{thm}\label{thm:conv}
Suppose that Assumption~\ref{ass:reg} holds. For each pair of given error thresholds $ \varepsilon_\obs > \varepsilon_\obs' > 0 $ in \eqref{eq:est-sw}, the estimation and optimization algorithms \eqref{eq:est-dyn} and \eqref{eq:ctrl-dyn} guarantee the following:
\begin{enumerate}
    \item There exists a time $ T \geq 0 $ such that
        \begin{equation}\label{eq:conv-obs-est}
            \|e_\obs(t)\| < \varepsilon_\obs, \quad \hat\theta(t) = \hat\theta(T) \qquad \forall t \geq T.
        \end{equation}
    \item The first-order necessary condition for optimality holds asymptotically for the optimization problem \eqref{eq:op-est-cost}, that is,
        \begin{equation}\label{eq:conv-ctrl}
            \lim_{t \to \infty} \big[ {-\jacob{r}} \hat J(r(t), \hat\theta(T))^\top \big]_{T_\cR(r(t))} = 0.
        \end{equation}
\end{enumerate}
\end{thm}

Essentially, item~1) ensures that the parameter estimate $ \hat\theta $ converges in finite time to a point that is indistinguishable from the actual value $ \theta $ using observations of the follower's action $ a = \hat f(\theta, r) $ and the leader's cost $ J(r, a) $, up to an error bounded by the threshold $ \varepsilon_\obs $. Regarding item~2), the necessity of \eqref{eq:conv-ctrl} for optimality is justified as follows.
\begin{lem}\label{lem:opt-nec}
If $ \hat r^* $ is a locally optimum of the optimization problem \eqref{eq:op-est-cost} with some fixed $ \hat\theta $, then
\begin{equation}\label{eq:opt-nec}
    \big[ {-\jacob{r}} \hat J(\hat r^*, \hat\theta)^\top \big]_{T_\cR(\hat r^*)} = 0.
\end{equation}
\end{lem}
\begin{proof}
The results in \cite[Th.~6.12, p.~207]{RockafellarWets1998} allow us to conclude that $ -\jacob{r} \hat J(\hat r^*, \hat\theta)^\top \in N_\cR(\hat r^*) $. Then \eqref{eq:opt-nec} follows from \eqref{eq:normal-cone}--\eqref{eq:tan-proj-zero}.
\end{proof}

\begin{proof}[Proof of Theorem~\ref{thm:conv}]
As the right hand-sides of \eqref{eq:est-dyn} and \eqref{eq:ctrl-dyn} are potentially discontinuous due to projections and switching, the proof of Theorem~\ref{thm:conv} utilizes results from differential inclusions theory; see Appendix~\ref{apx:proj-dyn} for the necessary preliminaries.

First, we establish the existence of solutions for the system defined by \eqref{eq:est-dyn} and \eqref{eq:ctrl-dyn}.
\begin{lem}\label{lem:dyn-soln}
For each $ (\hat\theta_0, r_0) \in \Theta \times \cR $, there exists a solution to the system defined by \eqref{eq:est-dyn} and \eqref{eq:ctrl-dyn} on $ \R_+ $---that is, there exist absolutely continuous functions $ \hat\theta: \R_+ \to \Theta $ and $ r: \R_+ \to \cR $ with $ (\hat\theta(0), r(0)) = (\hat\theta_0, r_0) $ such that \eqref{eq:est-dyn} and \eqref{eq:ctrl-dyn} hold almost everywhere on $ \R_+ $. Moreover, we have
\begin{equation}\label{eq:l-inf}
    \hat\theta, \dot{\hat\theta}, r, \dot r, e_\obs, \dot e_\obs \in \Linf.
\end{equation}
\end{lem}

\begin{proof}
Lemma~\ref{lem:dyn-soln} follows from results on hysteresis switching in \cite{MorseMayneGoodwin1992} and results on projected differential inclusions in \cite{Aubin1991}; see Appendix~\ref{apx:dyn-soln} for the complete proof.
\end{proof}

Second, we establish item~1) of Theorem~\ref{thm:conv} via arguments along the lines of the proof of Barbalat's lemma \cite[Lemma~3.2.6, p.~76]{IoannouSun1996}. We cannot use Barbalat's lemma directly since the switching signal $ \lambda_e $ in \eqref{eq:est-lya} is not continuous but only piecewise continuous. Following \eqref{eq:est-lya}, we see that $ \|\hat\theta - \theta\|^2 $ is monotonically decreasing. Then $ \lim_{t \to \infty} \|\hat\theta(t) - \theta\| $, and thus
\begin{equation}\label{eq:est-barbalat-int}
    \lim_{t \to \infty} \int_{0}^{t} \lambda_e(s) \|e_\obs(s)\|^2 \d s,
\end{equation}
exists and is finite. On the other hand, \eqref{eq:est-dyn} and \eqref{eq:est-sw} imply that \eqref{eq:conv-obs-est} holds if there exists a time $ T \geq 0 $ such that
\begin{equation}\label{eq:est-barbalat-sw}
    \lambda_e(t) = 0 \qquad \forall t \geq T.
\end{equation}
Assume \eqref{eq:est-barbalat-sw} does not hold for any $ T \geq 0 $. Then \eqref{eq:est-sw} implies that there exists an unbounded increasing sequence $ (t_k)_{k \in \N} $ with $ t_0 > 0 $ such that
\begin{equation}\label{eq:est-barbalat-sw-inv}
    \lambda_e(t_k) = \lambda_\theta, \quad \|e_\obs(t_k)\| > \varepsilon_\obs' \qquad \forall k \in \N.
\end{equation}
Next, we show that there exists an unbounded sequence $ (s_k)_{k \in \N} $ with $ s_k \in [t_k - \delta, t_k] $ such that
\begin{equation}\label{eq:est-barbalat-interval}
    \|e_\obs(t)\| > \varepsilon_\obs',\quad \lambda_e(t) = \lambda_\theta \qquad \forall k \in \N, \forall t \in [s_k, s_k + \delta)
\end{equation}
with the constant
\begin{equation*}
    \delta := \min \!\bigg\{ t_0,\, \frac{\varepsilon_\obs - \varepsilon_\obs'}{\sup_{s \geq 0} \|\dot e_\obs(s)\|} \bigg\} > 0,
\end{equation*}
where the inequality follows from $ t_0 > 0 $, $ \varepsilon_\obs > \varepsilon_\obs' $, and $ \dot e_\obs \in \Linf $ in \eqref{eq:l-inf}. Indeed, for each $ k \in \N $, consider the following two possibilities:
\begin{enumerate}
    \item If $ \|e_\obs(t)\| < \varepsilon_\obs $ for all $ t \in [t_k - \delta, t_k] $, then \eqref{eq:est-sw} and $ \lambda_e(t_k) = \lambda_\theta $ imply that \eqref{eq:est-barbalat-interval} holds with $ s_k = t_k - \delta $.
    \item Otherwise, there exists an $ s_k \in [t_k - \delta, t_k] $ such that $ \|e_\obs(s_k)\| = \varepsilon_\obs $, and \eqref{eq:est-barbalat-interval} follows from the definition of $ \delta $ and \eqref{eq:est-sw}.
\end{enumerate}
Moreover, $ (s_k)_{k \in \N} $ is unbounded as $ (t_k)_{k \in \N} $ is unbounded. Following \eqref{eq:est-barbalat-interval}, we have
\begin{equation*}
    \int_{s_k}^{s_k+\delta} \lambda_e(s) \|e_\obs(s)\|^2 \d s > \lambda_\theta (\varepsilon_\obs')^2 \delta > 0
\end{equation*}
for the unbounded sequence $ (s_k)_{k \in \N} $, which contradicts the property that \eqref{eq:est-barbalat-int} exists and is finite. Therefore, there exists a time $ T \geq 0 $ such that \eqref{eq:est-barbalat-sw}, and thus \eqref{eq:conv-obs-est}, holds.

Finally, we prove item~2) of Theorem~\ref{thm:conv} based on the invariance principle for projected gradient descent Proposition~\ref{prop:proj-dyn-lasalle} in Appendix~\ref{apx:proj-dyn}. After the time $ T $ from item~1), the system \eqref{eq:ctrl-dyn} becomes
\begin{equation*}
    \dot r = \big[ {-\lambda_r} \jacob{r} \hat J(r, \hat\theta(T))^\top
\big]_{T_\cR(r)},
\end{equation*}
which can be modeled using the projected dynamical system \eqref{eq:proj-dyn} in Appendix~\ref{apx:proj-dyn} with the state $ x := r $ and the set $ \cS := \cR $. The corresponding function $ g $ in \eqref{eq:proj-dyn} is given by
\begin{equation*}
    g(x) := -\lambda_r \jacob{r} \hat J(x, \hat\theta(T))^\top,
\end{equation*}
which satisfies \eqref{eq:proj-dyn-grad} with $ V(x) := \lambda_r \hat J(x, \hat\theta(T)) $. Then \eqref{eq:conv-ctrl} follows from \eqref{eq:proj-dyn-lasalle} in Proposition~\ref{prop:proj-dyn-lasalle}.
\end{proof}

In Theorem~\ref{thm:conv}, there is no claim that the parameter estimate $ \hat\theta $ necessarily converges to the actual value $ \theta $. However, this can be guaranteed if the following \emph{persistent excitation (PE)} condition holds.
\begin{ass}[Persistent excitation]\label{ass:pe}
There exist constants $ \tau_0, \alpha_0 > 0 $ such that the gain matrix $ K(r, a, \hat\theta) $ defined by \eqref{eq:est-gain} satisfies
\begin{equation}\label{eq:pe}
    \int_{t}^{t+\tau_0} K(s)^\top K(s) \d s \geq \alpha_0 I \qquad \forall t \geq 0,
\end{equation}
where we let $ K(s) := K(r(s), a(s), \hat\theta(s)) $ for brevity, and the inequality means that the difference of the left- and right-hand sides is a positive semidefinite matrix.
\end{ass}

\begin{thm}\label{thm:conv-pe}
Suppose that Assumptions~\ref{ass:reg} and~\ref{ass:pe} hold. For each given threshold $ \varepsilon_\theta > 0 $, by setting
\begin{equation}\label{eq:est-sw-pe}
    \varepsilon_\theta \sqrt{\alpha_0/\tau_0} \geq \varepsilon_\obs > \varepsilon_\obs' > 0
\end{equation}
in \eqref{eq:est-sw}, the estimation and optimization algorithms \eqref{eq:est-dyn} and \eqref{eq:ctrl-dyn} guarantee the following:
\begin{enumerate}
    \item There exists a time $ T \geq 0 $ such that \eqref{eq:conv-obs-est} holds and
        \begin{equation}\label{eq:conv-pe}
            \|\hat\theta(T) - \theta\| < \varepsilon_\theta.
        \end{equation}
    \item The first-order necessary condition for optimality holds asymptotically for the optimization problem \eqref{eq:op-est-cost}, that is, \eqref{eq:conv-ctrl} holds.
\end{enumerate}
\end{thm}
\begin{proof}
As \eqref{eq:conv-obs-est} and \eqref{eq:conv-ctrl} are established in Theorems~\ref{thm:conv}, it remains to prove \eqref{eq:conv-pe}. To this effect, we note that the inequality in \eqref{eq:conv-obs-est} implies
\begin{equation*}
    \int_{T}^{T + \tau_0} \|e_\obs(s)\|^2 \d s < \varepsilon_\obs^2 \tau_0 \leq \alpha_0 \varepsilon_\theta^2,
\end{equation*}
where the second inequality follows from \eqref{eq:est-sw-pe}. On the other hand, \eqref{eq:obs-est} and the equality in \eqref{eq:conv-obs-est} imply
\begin{equation*}
\begin{aligned}
    &\quad\, \int_{T}^{T + \tau_0} \|e_\obs(s)\|^2 \d s \\
    &= \int_{T}^{T + \tau_0} \|K(s) (\hat\theta(T) - \theta)\|^2 \d s \\
    &= (\hat\theta(T) - \theta)^\top \bigg( \int_{T}^{T + \tau_0} K(s)^\top K(s) \d s \Bigg) (\hat\theta(T) - \theta) \\
    &\geq \alpha_0 \|\hat\theta(T) - \theta\|^2,
\end{aligned}
\end{equation*}
where the inequality follows from the PE condition \eqref{eq:pe}. Combining the upper and lower bounds above yields \eqref{eq:conv-pe}.
\end{proof}

\begin{rmk}
In view of \eqref{eq:est-gain}, a sufficient condition for \eqref{eq:pe} is
\begin{equation}\label{eq:pe-sim}
    \int_{t}^{t+\tau_0} \jacob{\theta} f(r(s))^\top \jacob{\theta} f(r(s)) \d s \geq \alpha_0 I \qquad \forall t \geq 0.
\end{equation}
The PE condition \eqref{eq:pe-sim} is more restrictive than \eqref{eq:pe}; however, it can be checked without knowing the parameter estimate $ \hat\theta $.
\end{rmk}

\begin{rmk}\label{rmk:pe}
From the proof of Theorem~\ref{thm:conv-pe}, we see that \eqref{eq:conv-pe} only requires \eqref{eq:pe} or \eqref{eq:pe-sim} to hold at $ t = T $ for the time $ T $ from Theorem~\ref{thm:conv}. Therefore, to ensure \eqref{eq:conv-pe} in practice, it suffices to enforce \eqref{eq:pe} or \eqref{eq:pe-sim} when $ \lambda_e $ in \eqref{eq:est-sw} has been set to zero.
\end{rmk}

\section{Model mismatch}\label{sec:mismatch}
Up till now we assumed that there was some unknown value $ \theta $ from within the parameter set $ \Theta $ such that \eqref{eq:match} holds for the follower's strategy $ f $ and the parameterized function $ \hat f $ that is known to the leader. In this section, we consider the case where such perfect matching may not exist, and study the effect of a bounded mismatch between $ f(r) $ and $ \hat f(\theta, r) $.
\begin{ass}[Mismatch]\label{ass:mismatch-bnd}
The follower's strategy $ f $ is continuous, and there is an unknown value $ \theta \in \Theta $ such that
\begin{equation}\label{eq:mismatch-bnd}
    \|\hat f(\theta, r) - f(r)\| \leq \varepsilon_f/\kappa \qquad \forall r \in \cR
\end{equation}
with the constant
\begin{equation*}
    \kappa := \max_{r \in \cR,\, \hat\theta \in \Theta} \bigg\| \!\begin{bmatrix}
        I \\
        \int_{0}^{1} \jacob{a} J \big( r, \rho \hat f(\hat\theta, r) + (1 - \rho) f(r) \big) \d\rho
    \end{bmatrix}\! \bigg\|
\end{equation*}
and some known constant $ \varepsilon_f \geq 0 $.
\end{ass}

The follower's action set $ \cA $ is still unknown to the leader, except that $ \cA \subset \cup_{\hat\theta \in \Theta,\, r \in \cR} (\varepsilon_f/\kappa)\, \bB \big( \hat f(\hat\theta, r) \big) $ as implied by \eqref{eq:mismatch-bnd}. Assumption~\ref{ass:mismatch-bnd} generalized the condition \eqref{eq:match} as \eqref{eq:match} is equivalent to \eqref{eq:mismatch-bnd} with $ \varepsilon_f = 0 $.

Similar arguments to those in the proof of Lemma~\ref{lem:obs-est} show that the observation error $ e_\obs $ now satisfies
\begin{equation}\label{eq:obs-est-mismatch}
    e_\obs = K(r, a, \hat\theta) (\hat\theta - \theta) + e_f
\end{equation}
with the gain matrix $ K(r, a, \hat\theta) $ defined by \eqref{eq:est-gain} and the \emph{mismatch error}
\begin{equation*}
    e_f := \begin{bmatrix}
        I \\
        \int_{0}^{1} \jacob{a} J \big( r, \rho \hat f(\hat\theta, r) + (1 - \rho)\, a \big) \d\rho
    \end{bmatrix} \big( \hat f(\theta, r) - a \big).
\end{equation*}
It turns out that Assumption~\ref{ass:mismatch-bnd} guarantees
\begin{equation}\label{eq:mismatch-ineq}
    \|e_f(t)\| \leq \varepsilon_f \qquad \forall t \geq 0.
\end{equation}

The following two results extend Theorems~\ref{thm:conv} and~\ref{thm:conv-pe} to the current case without perfect matching between $ f(r) $ and $ \hat f(\theta, r) $ for some $ \theta \in \Theta $.

\begin{thm}\label{thm:conv-mismatch}
Suppose that Assumptions~\ref{ass:reg} and~\ref{ass:mismatch-bnd} hold. For each pair of given thresholds $ \varepsilon_\obs > \varepsilon_\obs' > \varepsilon_f $ in \eqref{eq:est-sw}, the estimation and optimization algorithms \eqref{eq:est-dyn} and \eqref{eq:ctrl-dyn} guarantee the same convergence results as those in Theorem~\ref{thm:conv}.
\end{thm}
\begin{proof}
First, Lemma~\ref{lem:dyn-soln} still holds as $ f $ is continuous.

Second, we establish item~1) of Theorem~\ref{thm:conv-mismatch} via similar arguments to those in Section~\ref{ssec:alg-est} and the second step of the proof of Theorem~\ref{thm:conv}. Using the estimation algorithm \eqref{eq:est-dyn} with $ \varepsilon_\obs > \varepsilon_\obs' > \varepsilon_f $ in \eqref{eq:est-sw}, the estimation error $ \hat\theta - \theta $ now satisfies
\begin{equation*}
\begin{aligned}
    \frac{\d\|\hat\theta - \theta\|^2}{\d t} &= 2 (\hat\theta - \theta)^\top \big[ {-\lambda_e} K(r, a, \hat\theta)^\top e_\obs \big]_{T_\Theta(\hat\theta)} \\
    &\leq 2 (\hat\theta - \theta)^\top \big( {-\lambda_e} K(r, a, \hat\theta)^\top e_\obs \big) \\
    &= -2 \lambda_e (e_\obs - e_f)^\top e_\obs,
\end{aligned}
\end{equation*}
where the inequality follows from \eqref{eq:tan-cone}--\eqref{eq:tan-proj-cone}. Next, we prove that
\begin{equation}\label{eq:est-lya-mismatch}
    \frac{\d\|\hat\theta - \theta\|^2}{\d t} \leq -2 \lambda_e (e_\obs - e_f)^\top e_\obs \leq 0,
\end{equation}
in which $ \d\|\hat\theta - \theta\|^2/\d t = 0 $ if and only if $ \lambda_e = 0 $. Indeed, consider the following two possibilities:
\begin{enumerate}
    \item If $ \lambda_e = 0 $, then $ \d\|\hat\theta - \theta\|^2/\d t = 0 $.
    \item Otherwise $ \lambda_e = \lambda_\theta $, and thus $ \|e_\obs\| > \varepsilon_\obs' > \varepsilon_f \geq \|e_f\| $ following \eqref{eq:est-sw} and \eqref{eq:mismatch-ineq}. Hence
        \begin{equation*}
        \begin{aligned}
            \frac{\d\|\hat\theta - \theta\|^2}{\d t} &\leq -2 \lambda_\theta (e_\obs - e_f)^\top e_\obs \\
            &\leq -2 \lambda_\theta (\|e_\obs\|^2 - \|e_f\| \|e_\obs\|) \\
            &= -2 \lambda_\theta (\varepsilon_\obs' - \varepsilon_f)\, \varepsilon_\obs' < 0.
        \end{aligned}
        \end{equation*}
\end{enumerate}
Following \eqref{eq:est-lya-mismatch}, we see that $ \|\hat\theta - \theta\|^2 $ is monotonically decreasing. Thus $ \lim_{t \to \infty} \|\hat\theta(t) - \theta\|^2 $, and therefore
\begin{equation}\label{eq:est-barbalat-int-mismatch}
    \lim_{t \to \infty} \int_{0}^{t} \lambda_e(s) (e_\obs(s) - e_f(s))^\top e_\obs(s) \d s,
\end{equation}
exists and is finite. On the other hand, \eqref{eq:est-dyn} and \eqref{eq:est-sw} imply that \eqref{eq:conv-obs-est} holds if there exists a time $ T \geq 0 $ such that \eqref{eq:est-barbalat-sw} holds. Assume \eqref{eq:est-barbalat-sw} does not hold for any $ T \geq 0 $. Then the analysis in the second step of the proof of Theorem~\ref{thm:conv} shows that there exists an unbounded sequence $ (s_k)_{k \in \N} $ with $ s_k \in [t_k - \delta, t_k] $ such that \eqref{eq:est-barbalat-interval} holds. Consequently, we have
\begin{equation*}
\begin{aligned}
    &\quad\, \int_{s_k}^{s_k+\delta} \lambda_e(s) (e_\obs(s) - e_f(s))^\top e_\obs(s) \d s \\
    &\geq \lambda_\theta \int_{s_k}^{s_k+\delta} (\|e_\obs(s)\|^2 - \|e_f(s)\| \|e_\obs(s)\|) \d s \\
    &> \lambda_\theta (\varepsilon_\obs' - \varepsilon_f)\, \varepsilon_\obs' \delta > 0
\end{aligned}
\end{equation*}
for the unbounded sequence $ (s_k)_{k \in \N} $, which, combined with \eqref{eq:est-lya-mismatch}, contradicts the property that \eqref{eq:est-barbalat-int-mismatch} exists and is finite. Therefore, there exists a time $ T \geq 0 $ such that \eqref{eq:est-barbalat-sw}, and thus \eqref{eq:conv-obs-est}, holds.

Finally, item~2) of Theorem~\ref{thm:conv-mismatch} is the same as item~2) of Theorem~\ref{thm:conv} as the optimization process is the same after the adaptation of $ \hat\theta $ stops.
\end{proof}

\begin{thm}\label{thm:conv-pe-mismatch}
Suppose that Assumptions~\ref{ass:reg}--\ref{ass:mismatch-bnd} hold. For each given threshold $ \varepsilon_\theta > 2 \varepsilon_f \sqrt{\tau_0/\alpha_0} $, by setting
\begin{equation}\label{eq:est-sw-pe-mismatch}
    \varepsilon_\theta \sqrt{\alpha_0/\tau_0} - \varepsilon_f \geq \varepsilon_\obs > \varepsilon_\obs' > \varepsilon_f
\end{equation}
in \eqref{eq:est-sw}, the estimation and optimization algorithms \eqref{eq:est-dyn} and \eqref{eq:ctrl-dyn} guarantee the same convergence results as those in Theorem~\ref{thm:conv-pe}.
\end{thm}
\begin{proof}
As \eqref{eq:conv-obs-est} and \eqref{eq:conv-ctrl} are established in Theorems~\ref{thm:conv-mismatch}, it remains to prove \eqref{eq:conv-pe}. To this effect, we note that the inequality in \eqref{eq:conv-obs-est} implies
\begin{equation*}
    \int_{T}^{T + \tau_0} \|e_\obs(s)\|^2 \d s < \varepsilon_\obs^2 \tau_0 \leq (\varepsilon_\theta \sqrt{\alpha_0} - \varepsilon_f \sqrt{\tau_0})^2,
\end{equation*}
where the second inequality follows from \eqref{eq:est-sw-pe-mismatch}. On the other hand, \eqref{eq:obs-est-mismatch} and the equality in \eqref{eq:conv-obs-est} imply
\begin{equation*}
\begin{aligned}
    &\quad\,\, \|e_\obs(s)\|^2 \\
    &= \|K(s) (\hat\theta(T) - \theta) + e_f(s)\|^2 \\
    &\geq \bigg( 1 - \frac{\varepsilon_f}{\varepsilon_\theta} \sqrt{\frac{\tau_0}{\alpha_0}} \bigg) \|K(s) (\hat\theta(T) - \theta)\|^2 \\
    &\quad\, + \bigg( 1 - \frac{\varepsilon_\theta}{\varepsilon_f} \sqrt{\frac{\alpha_0}{\tau_0}} \bigg) \|e_f(s)\|^2 \\
    &= (\varepsilon_\theta \sqrt{\alpha_0} - \varepsilon_f \sqrt{\tau_0}) \bigg( \frac{\|K(s) (\hat\theta(T) - \theta)\|^2}{\varepsilon_\theta \sqrt{\alpha_0}} - \frac{\|e_f(s)\|^2}{\varepsilon_f \sqrt{\tau_0}} \bigg)
\end{aligned}
\end{equation*}
for all $ s \geq T $, where the inequality follows from Young's inequality $ 	2 a b \leq \varepsilon a^2 + b^2/\varepsilon $ for all $ a, b \in \R $ and $ \varepsilon > 0 $. Note that $ \varepsilon_\theta > 2 \varepsilon_f \sqrt{\tau_0/\alpha_0} $ implies $ \varepsilon_\theta \sqrt{\alpha_0} - \varepsilon_f \sqrt{\tau_0} > 0 $. Hence we have
\begin{equation*}
\begin{aligned}
    &\quad\, \int_{T}^{T + \tau_0} \|e_\obs(s)\|^2 \d s \\
    &\geq (\varepsilon_\theta \sqrt{\alpha_0} - \varepsilon_f \sqrt{\tau_0}) \bigg( \int_{T}^{T + \tau_0} \frac{\|K(s) (\hat\theta(T) - \theta)\|^2}{\varepsilon_\theta \sqrt{\alpha_0}} \d s \\
    &\quad\, - \int_{T}^{T + \tau_0} \frac{\|e_f(s)\|^2}{\varepsilon_f \sqrt{\tau_0}} \d s \bigg) \\
    &\geq (\varepsilon_\theta \sqrt{\alpha_0} - \varepsilon_f \sqrt{\tau_0}) \bigg( \frac{\sqrt{\alpha_0}}{\varepsilon_\theta} \|\hat\theta(T) - \theta\|^2 - \varepsilon_f \sqrt{\tau_0} \bigg),
\end{aligned}
\end{equation*}
where the second inequality follows from the PE condition \eqref{eq:pe} and the inequality \eqref{eq:mismatch-ineq}. Combining the upper and lower bounds above yields \eqref{eq:conv-pe}.
\end{proof}

\begin{rmk}\label{rmk:mismatch}
The results in this section also hold with the slightly less conservative condition
\begin{multline}\label{eq:mismatch-bnd-new}
    \min_{\theta \in \Theta} \max_{r \in \cR,\, \hat\theta \in \Theta} \bigg\| \!\begin{bmatrix} I \\ \int_{0}^{1} \jacob{a} J \big( r, \rho \hat f(\hat\theta, r) + (1 - \rho) f(r) \big) \d\rho \end{bmatrix} \\
    \big( \hat f(\theta, r) - f(r) \big) \bigg\| \leq \varepsilon_f
\end{multline}
in place of \eqref{eq:mismatch-bnd} in Assumption~\ref{ass:mismatch-bnd}.
\end{rmk}

\section{Simulation example}\label{sec:sim}
We illustrate the estimation and optimization algorithms and the convergence results via a simulation example motivated by link-flooding distributed denial-of-service (DDoS) attacks, such as the Crossfire attack \cite{KangLeeGligor2013}.

Consider a communication network consisting of $ L $ parallel links connecting a source to a destination. The set of links is denoted by $ \cL := \{1,\, \ldots, L\} $. Suppose that a router (leader) distributes $ R $ units of legitimate traffic among the parallel links, and an attacker (follower) disrupts communication by injecting $ A $ units of superfluous traffic on them. The router's action is represented by an $ L $-vector of the \emph{desired} legitimate traffic on each link $ r \in \cR := \{r \in \R_+^L: \sum_{l \in \cL} r_l = R\} $, and the attacker's action is represented by an $ L $-vector of the attack traffic $ a \in \cA := \{a \in \R_+^L: \sum_{l \in \cL} a_l = A\} $. Every link $ l \in \cL $ is subject to a constant capacity $ c_0 > 0 $ that upper-bounds the total traffic on $ l $. When $ r_l + a_l > c_0 $, the \emph{actual} legitimate traffic on $ l $ is decreased to
\begin{equation*}
    u_l := \min\{r_l,\, \max\{c_0 - a_l,\, 0\}\}.
\end{equation*}
The router aims to maximize the total actual legitimate traffic, whereas the attacker aims to minimize it. Hence the router's cost is
\begin{equation*}
    J(r, a) := -\sum_{l \in \cL} u_l,
\end{equation*}
and the attacker's cost is
\begin{equation}\label{eq:sim-atk-cost}
    H(a, r) := \sum_{l \in \cL} u_l = -J(r, a).
\end{equation}
Clearly, neither the router nor the attacker has an incentive to assign more traffic on a link than the capacity. Hence we assume $ r_l, a_l \in [0, c_0] $ for all $ l \in \cL $. For most nontrivial cases, the game defined by $ (\cR, \cA, J, H) $ has no Nash equilibrium. On the other hand, in \cite[Cor.~5]{YangHespanha2021} it was established that there exists a Stackelberg equilibrium action given by $ r^*_l = R/L $ for all $ l \in \cL $.

If the router knew that the attacker's cost function was indeed defined by \eqref{eq:sim-atk-cost}, it could play the Stackelberg equilibrium action $ r^* $. However, we consider the general scenario where it does not and, instead, adopts the adaptive learning approach proposed in this paper to construct its optimal action. To this effect, the router approximates the attacker's strategy $ f = (f_1, \ldots, f_L) $ using a quasi-RBF model defined by\footnote{The maps $ r \mapsto J(r, a) $ and $ r \mapsto \hat f(\hat\theta, r) $ in this example actually violate the smoothness conditions in Assumption~\ref{ass:reg} as they are only piecewise continuously differentiable. However, these conditions are only needed so that the optimization algorithm \eqref{eq:ctrl-dyn} is well-defined and does not lead to chattering. In this example, the set of non-differentiable points in $ \cR $ has measure zero and does not affect the simulation.\label{ftnf:sim-cont}}
\begin{multline}
    \hat f_l(\theta, r) := \sum_{j_1=1}^{n_\rbf} \cdots \sum_{j_{L-1}=1}^{n_\rbf} \theta_{l, j_1, \ldots, j_{L-1}} F_{j_1, \ldots, j_{L-1}}(r) \\
    = \sum_{j_1=1}^{n_\rbf} \cdots \sum_{j_{L-1}=1}^{n_\rbf} \theta_{l, j_1, \ldots, j_{L-1}} \phi(r - r^c_{j_1, \ldots, j_{L-1}}) \label{eq:sim-rbf}
\end{multline}
for $ l \in \cL $, where the RBF is defined by
\begin{equation*}
    \phi(x) := \indfcn_{\big( {-\frac{c_0}{2 n_\rbf}}, \frac{c_0}{2 n_\rbf} \big]^{L-1}}(x)
\end{equation*}
with the indicator function
\begin{equation*}
    \indfcn_\cS(x) := \begin{cases}
        1 &\text{if } x \in \cS; \\
        0 &\text{if } x \notin \cS,
    \end{cases}
\end{equation*}
the centers are defined by
\begin{equation*}
    r^c_{j_1, \ldots, j_{L-1}} := \bigg( \frac{(2 j_1 - 1)\, c_0}{2 m}, \ldots, \frac{(2 j_{L-1} - 1)\, c_0}{2 m} \bigg),
\end{equation*}
and $ n_\rbf \in \N $ is the number of kernels in each of the first $ L - 1 $ scalar component of $ \cR $, that is, the router approximates each scalar component of $ f $ using a grid of $ n_\rbf^{L-1} $ hypercubes. Hence $ n_\theta := L n_\rbf^{L-1} $ and $ \Theta := [0, c_0]^{n_\theta} $ in \eqref{eq:match}. In the following, we simulate our estimation and learning algorithms \eqref{eq:est-dyn} and \eqref{eq:ctrl-dyn} for networks with two or three parallel links. In these simulations, the constants are set by $ \varepsilon_\obs = 2 \varepsilon_\obs' = 0.002 $, $ \lambda_\theta = 0.02 $, and $ \lambda_r = 0.002 $, and the initial values of the parameter estimate $ \hat\theta $ and the routers' action $ r $ are randomly generated.

\subsection{Network with two parallel links}
\begin{figure}[!htbp]
\centering
\includegraphics[draft=false,width=1\columnwidth,max width=128pt,trim=.5em 0ex 0em 2ex,clip]{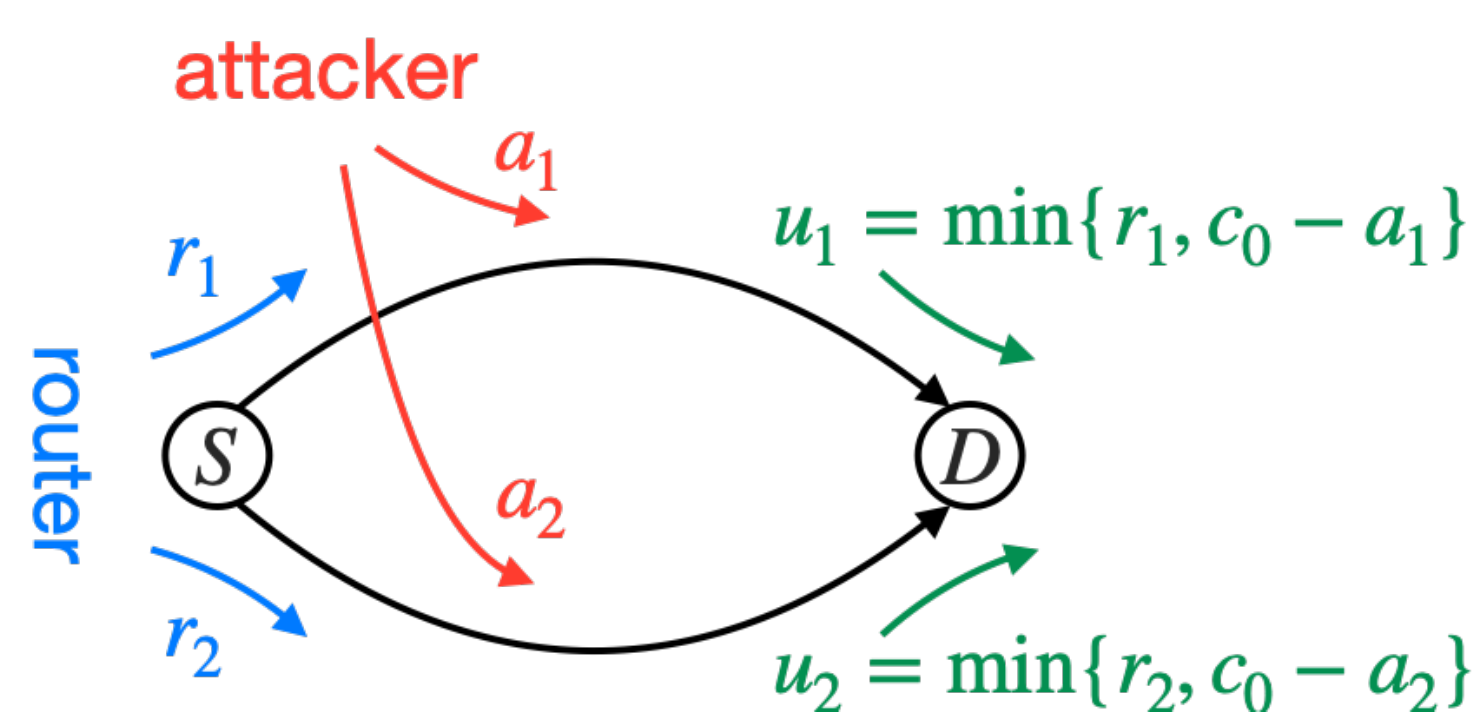}%
\caption{A network with one source $ S $, one destination $ D $, and two parallel links (assuming that $ a_l \leq c_0 $ for all $ l \in \cL $).}\label{fig:eg-net-2}
\end{figure}

Consider the network with $ L = 2 $ parallel links in Fig.~\ref{fig:eg-net-2}, capacity $ c_0 = 1 $, total desired legitimated traffic $ R = L c_0/2 = 1 $, and attack budget $ A = \lceil L c_0/2 \rceil = 1 $. We set the constant $ n_\rbf = 4 $. Then $ n_\theta = 8 $ and the parameter set $ \Theta = [0, 1]^8 $. Following \cite[Cor.~2]{YangHespanha2021}, the attacker's best response to a router action $ r $ is to set $ a_l = 1 $ on the link $ l \in \{1,\, 2\} $ with larger $ r_l $. Specifically, the actual value $ \theta $ in \eqref{eq:match}, written as the tensor form in \eqref{eq:sim-rbf}, is given by
\begin{equation*}
    \theta = \begin{bmatrix}
        0 & 0 & 1 & 1 \\
        1 & 1 & 0 & 0
    \end{bmatrix}.
\end{equation*}
As established in \cite[Cor.~5]{YangHespanha2021}, the Stackelberg equilibrium action is given by $ r^* = (0.5, 0.5) $.

\begin{figure}[!htbp]
\centering
\subfloat[Observation error]{\includegraphics[width=.49\columnwidth,max width=128pt,trim=1.5em 4ex 1.5em 3.5ex,clip]{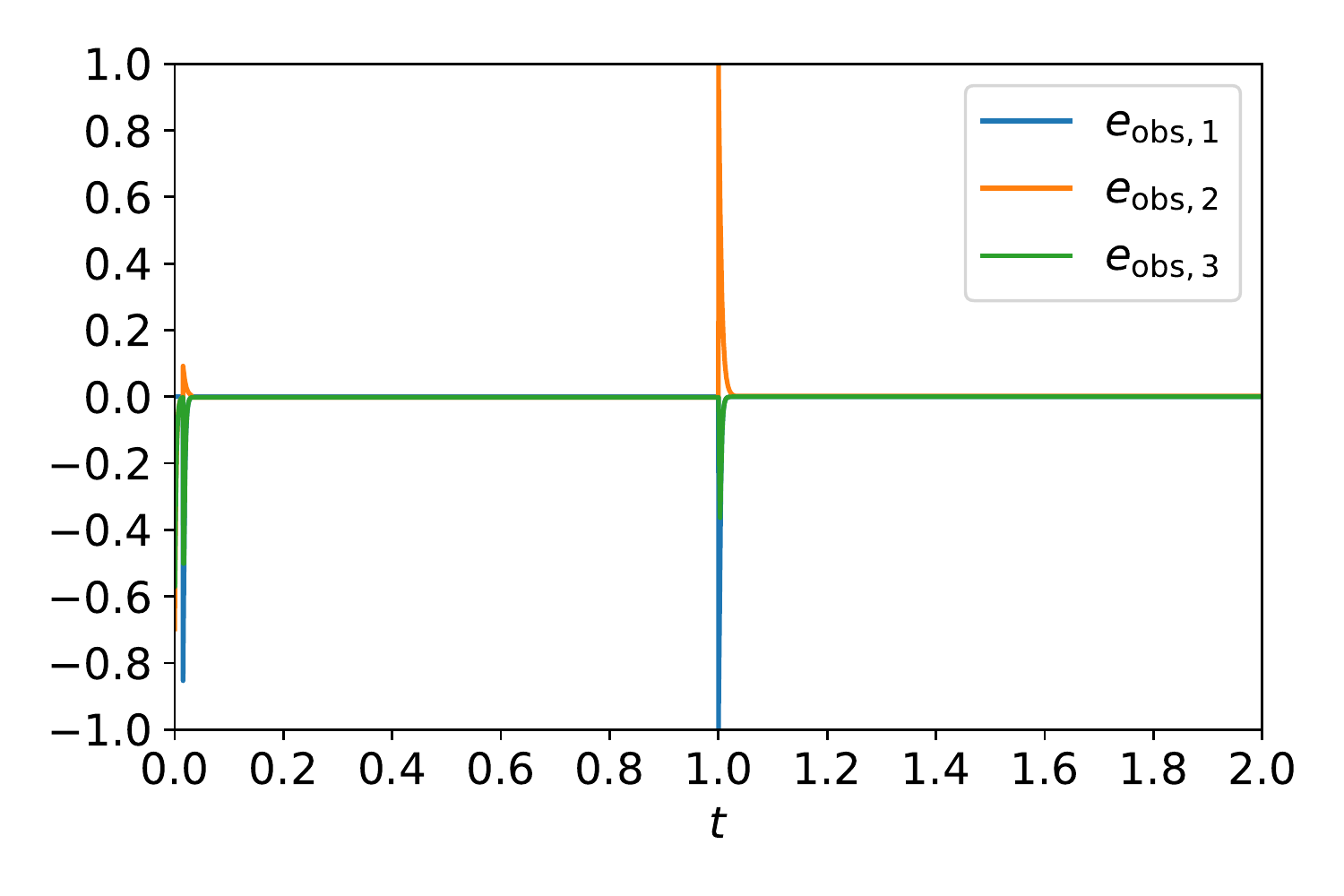}\label{fig:eg-2-obs}}%
\hfill%
\subfloat[Router's action]{\includegraphics[width=.49\columnwidth,max width=128pt,trim=1.5em 4ex 1.5em 3.5ex,clip]{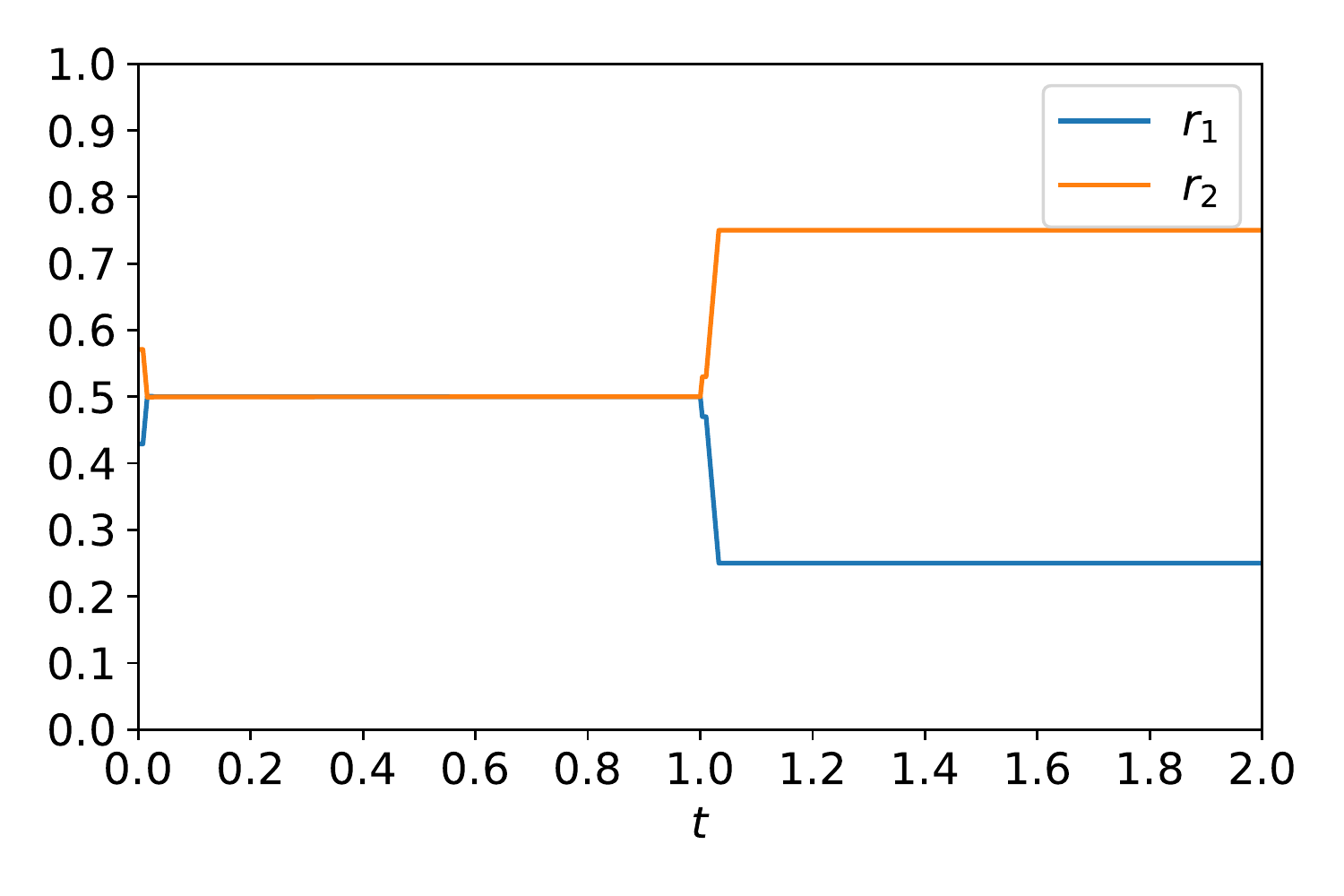}\label{fig:eg-2-rtg}}%
\\%
\subfloat[Router's actual and predicted costs]{\includegraphics[width=.49\columnwidth,max width=128pt,trim=1.5em 4ex 1.5em 3.5ex,clip]{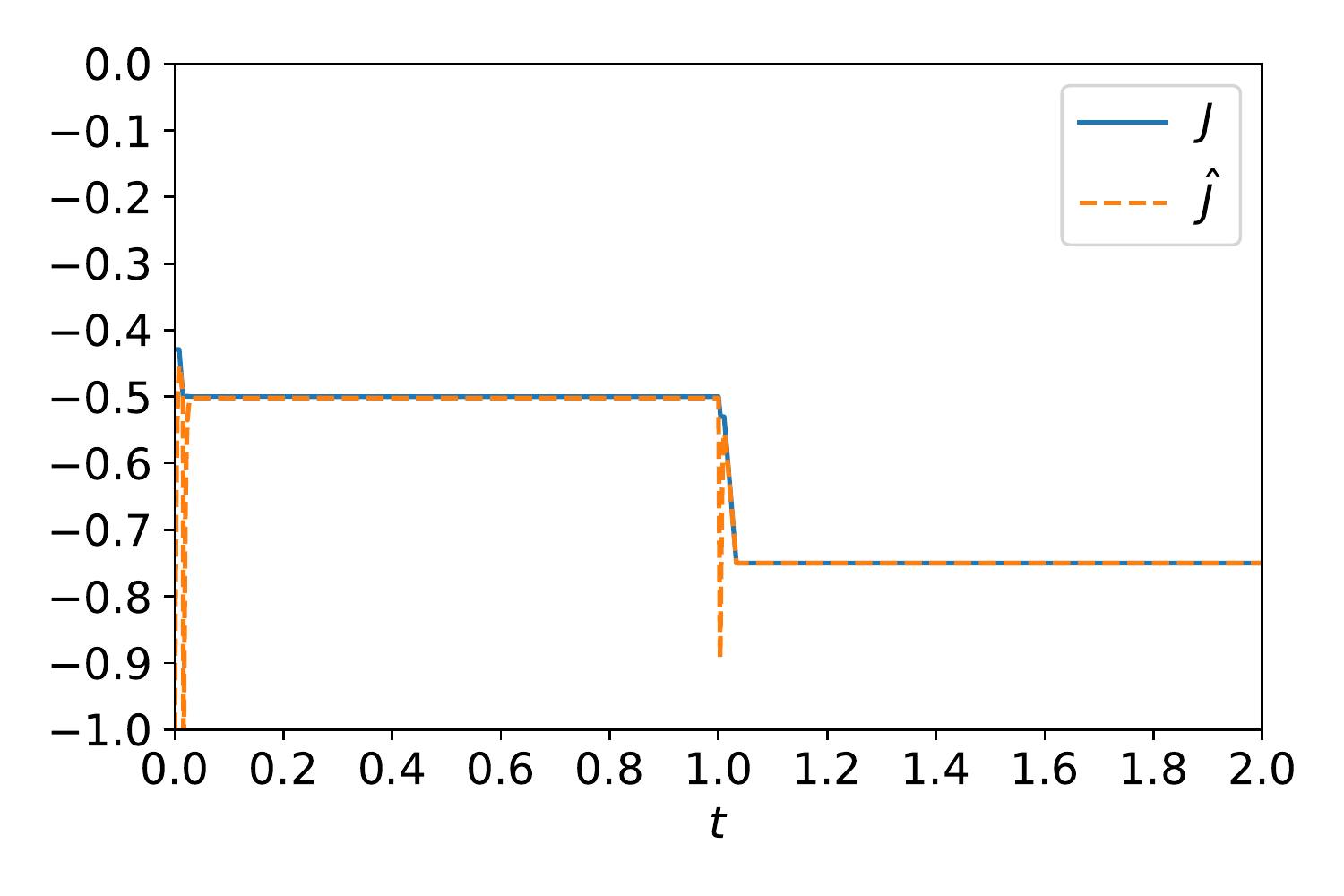}\label{fig:eg-2-rtg-cost}}%
\hfill%
\subfloat[Attacker's cost]{\includegraphics[width=.49\columnwidth,max width=128pt,trim=1.5em 4ex 1.5em 3.5ex,clip]{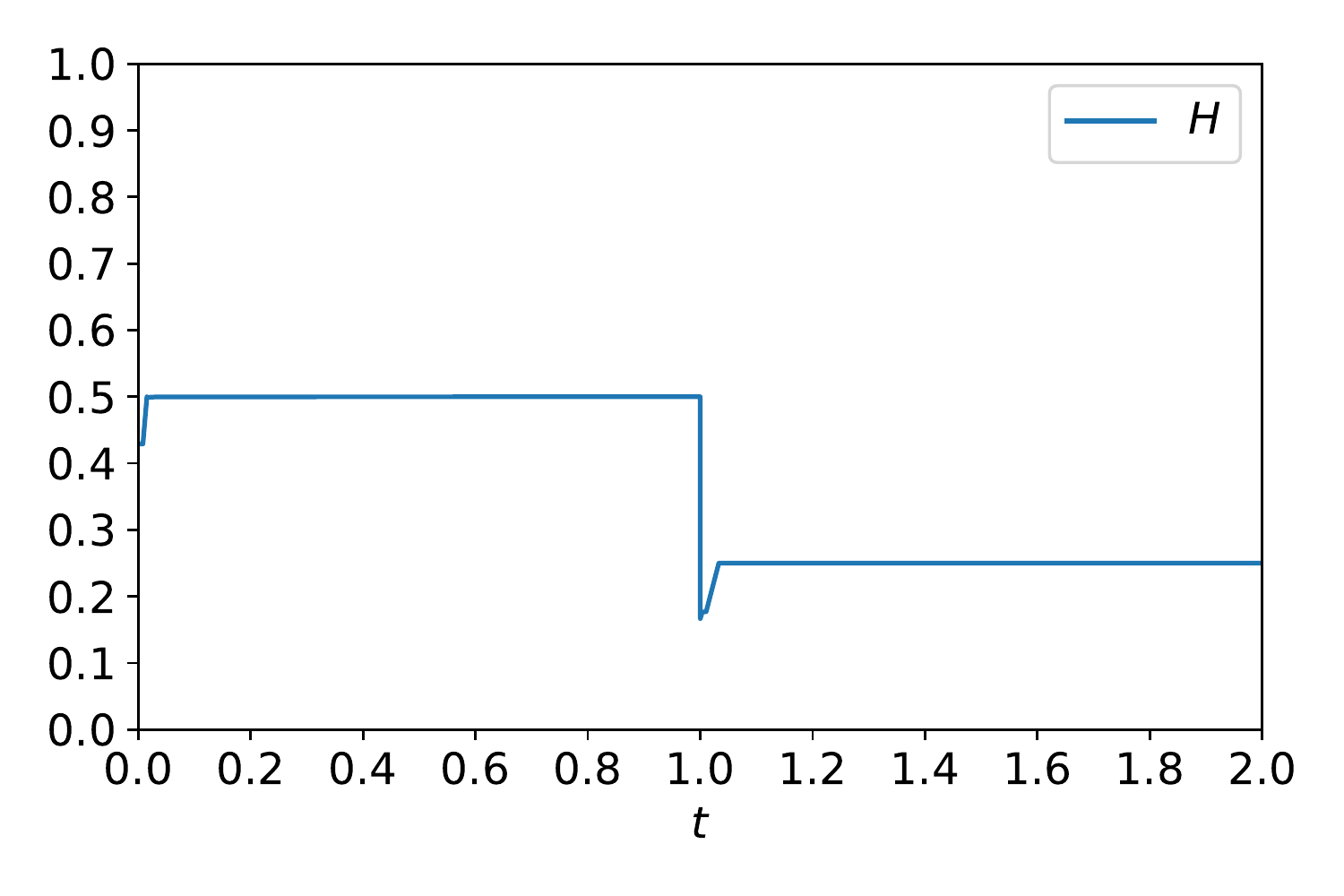}\label{fig:eg-2-atk-cost}}%
\caption{Simulation results for $ L = 2 $ w/o PE (horizontal axis: $ \times 10^4 $ units of time). In the first $ 10^4 $ units of time, the observation error converges to $ e_\obs = 0 $, the router's action $ r $ converges to the Stackelberg equilibrium action $ r^* = (0.5, 0.5) $, the router's and attacker's costs converge to $ J = \hat J = -H = -0.5 $; in the second $ 10^4 $ units of time, the attacker switches to a new cost function, the router's action $ r $ converges to an $ \varepsilon $ Stackelberg action near $ \bar r^* = (0.25, 0.75) $, the router's actual and predicted costs converge to $ J = \hat J = -0.75 $, while the attacker's cost converges to $ H = 0.25 $.}\label{fig:eg-2}%
\end{figure}

\begin{figure}[!htbp]
\centering
\subfloat[Router's actual and predicted cost functions at $ T = 10^4 $]{\includegraphics[width=.49\columnwidth,max width=128pt,trim=1.5em 3.5ex 1.5em 3.5ex,clip]{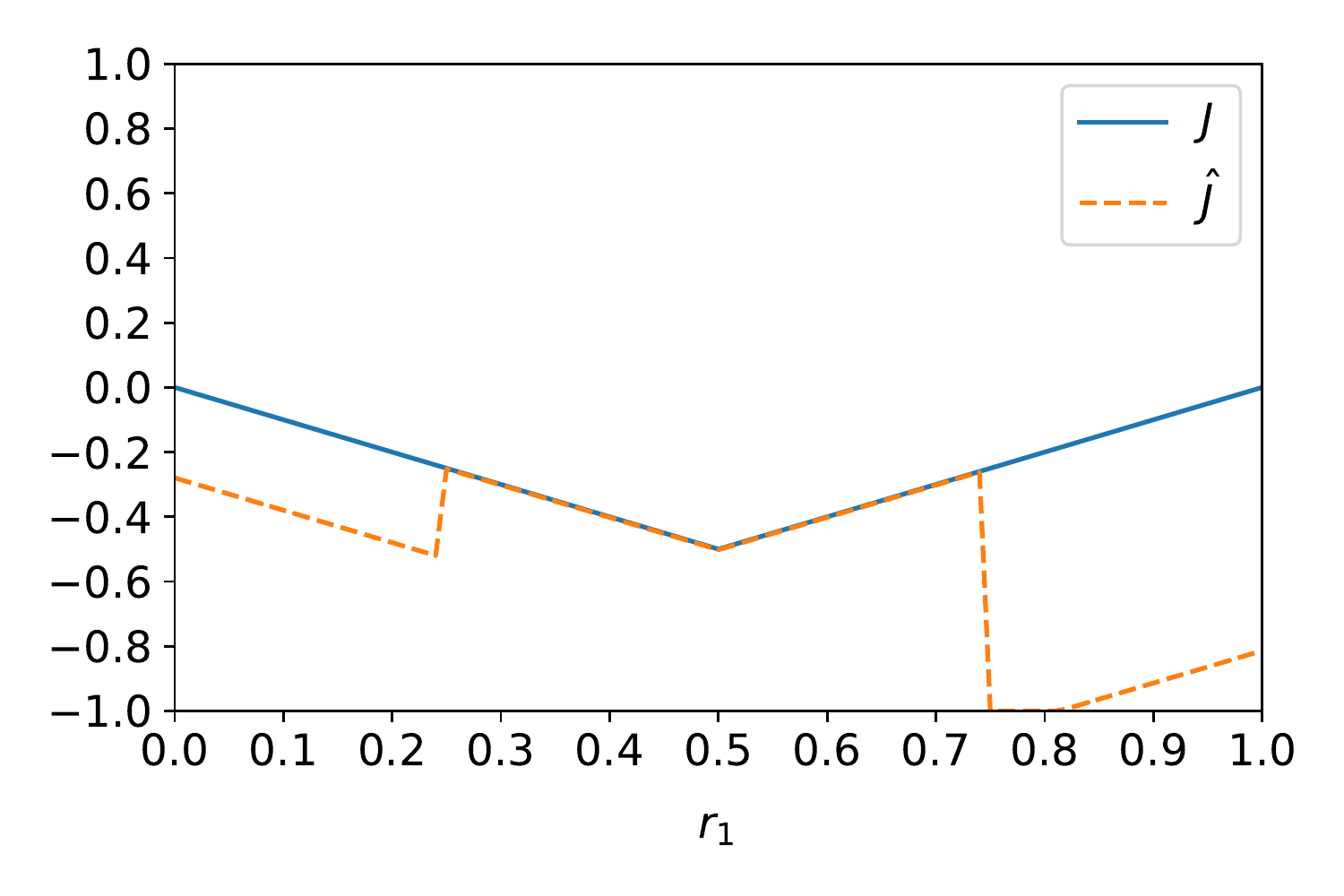}\label{fig:eg-2-fcn-1}}%
\hfill%
\subfloat[Router's actual and predicted cost functions at $ T = 2 \times 10^4 $]{\includegraphics[width=.49\columnwidth,max width=128pt,trim=1.5em 3.5ex 1.5em 3.5ex,clip]{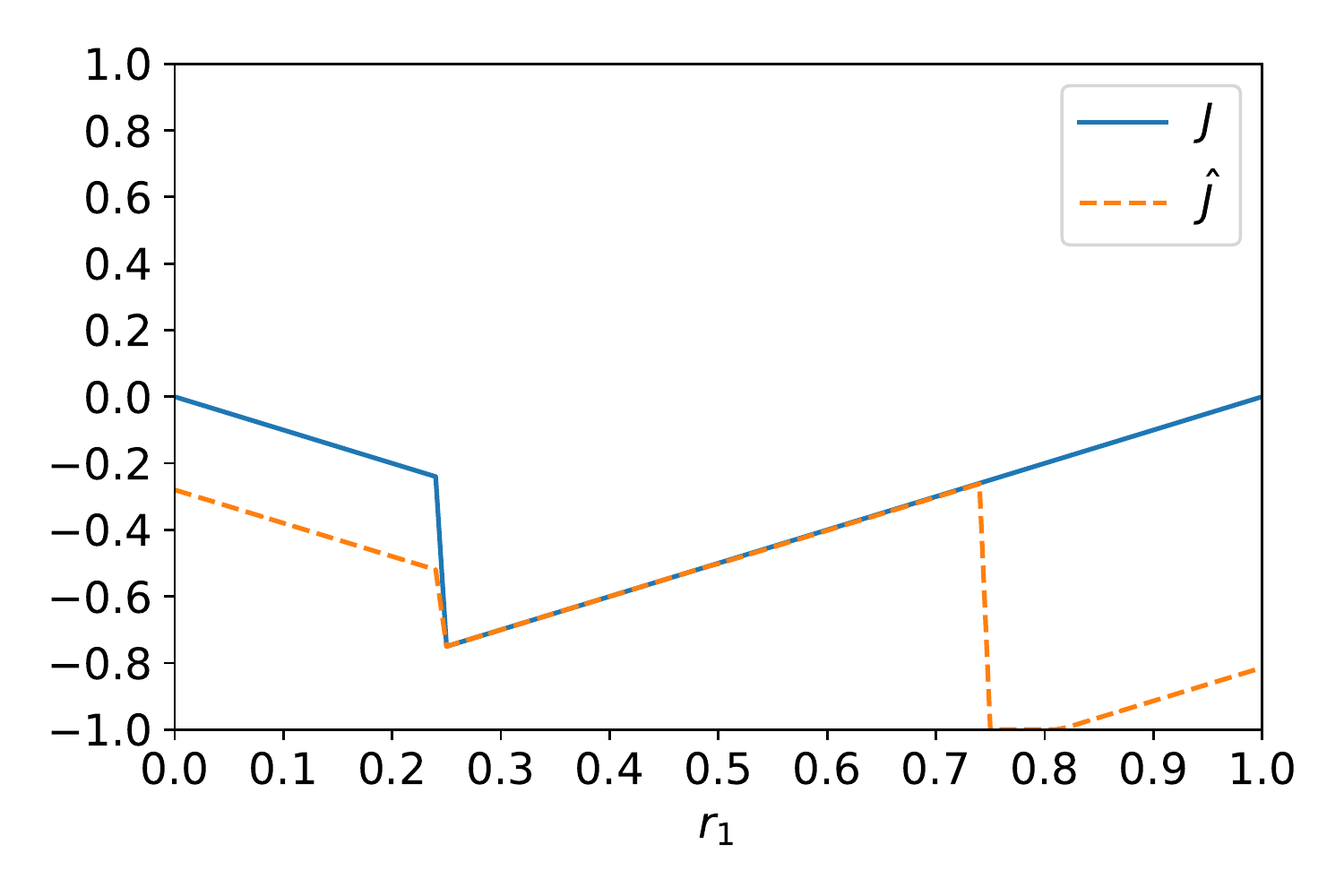}\label{fig:eg-2-fcn-2}}%
\caption{Router's actual cost function $ r_1 \mapsto J(r, f(r)) $ and predicted cost function $ r_1 \mapsto \hat J(\hat\theta(T), r) $ for $ L = 2 $ w/o PE. In both scenarios, the predicted cost function is accurate near the optima $ r^* $ but not everywhere.}\label{fig:eg-2-fcn}%
\end{figure}

\begin{figure}[!htbp]
\centering
\subfloat[Observation error]{\includegraphics[width=.49\columnwidth,max width=128pt,trim=1.5em 4ex 1.5em 3.5ex,clip]{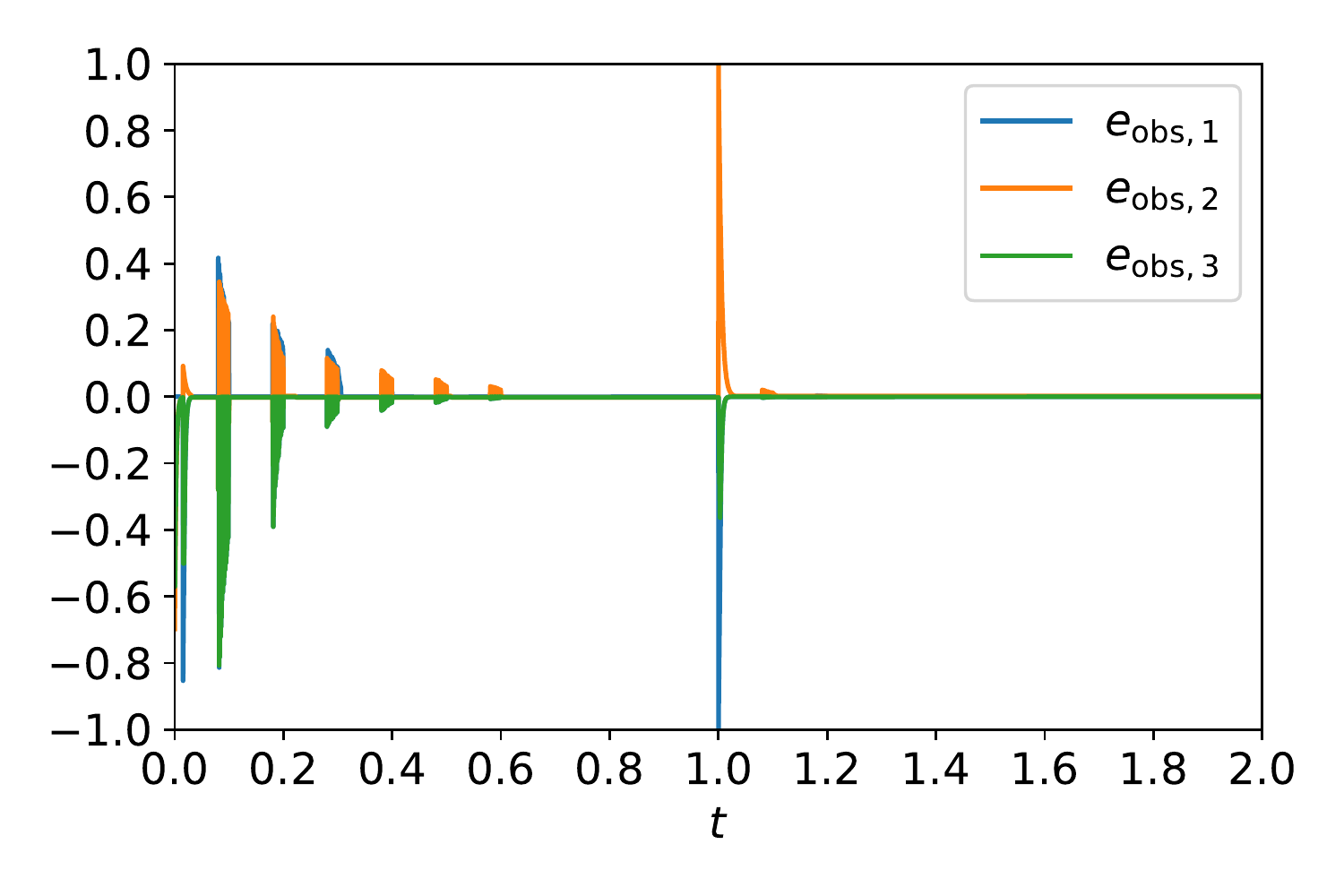}\label{fig:eg-2-obs-pe}}%
\hfill%
\subfloat[Router's action]{\includegraphics[width=.49\columnwidth,max width=128pt,trim=1.5em 4ex 1.5em 3.5ex,clip]{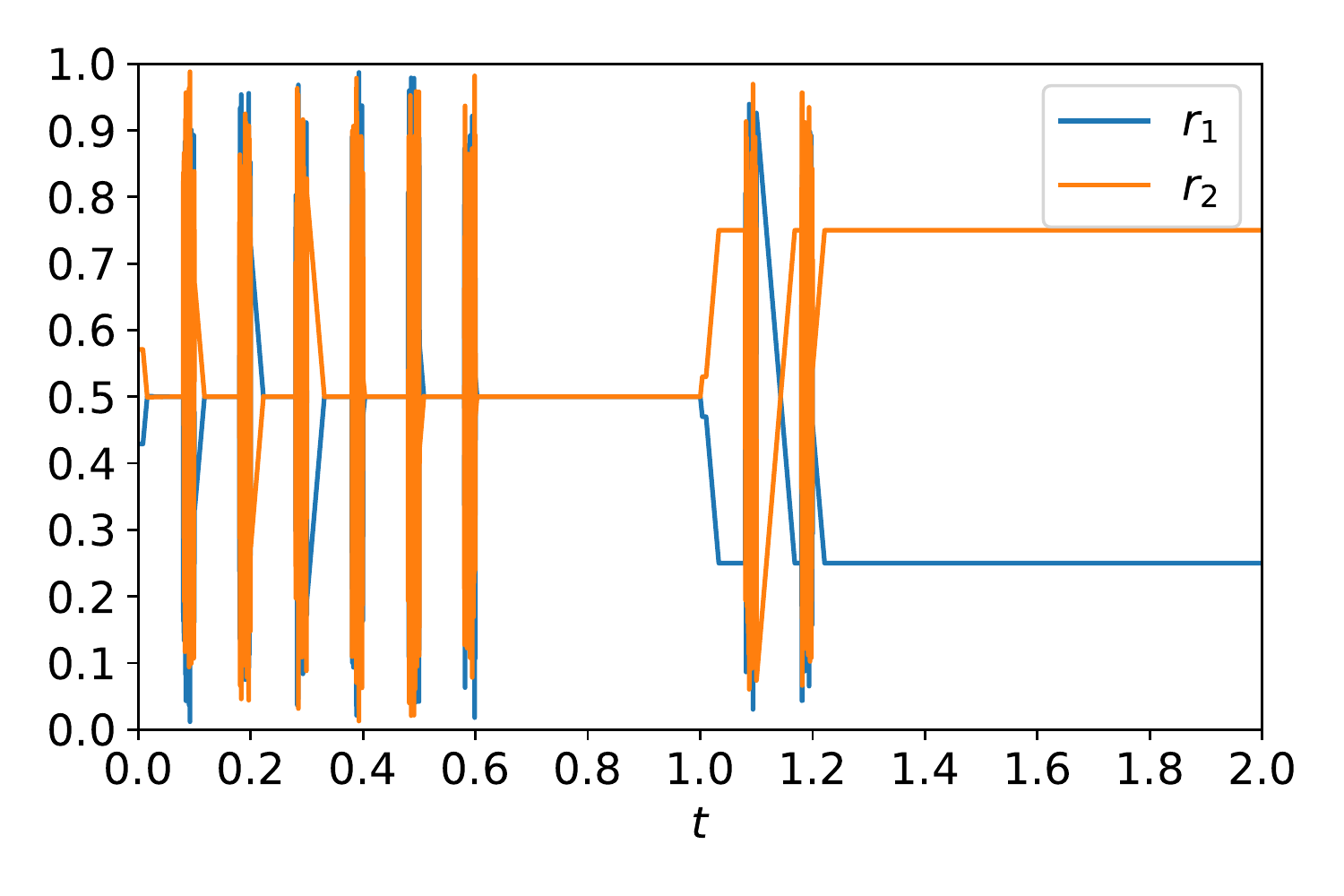}\label{fig:eg-2-rtg-pe}}%
\\%
\subfloat[Router's actual and predicted costs]{\includegraphics[width=.49\columnwidth,max width=128pt,trim=1.5em 4ex 1.5em 3.5ex,clip]{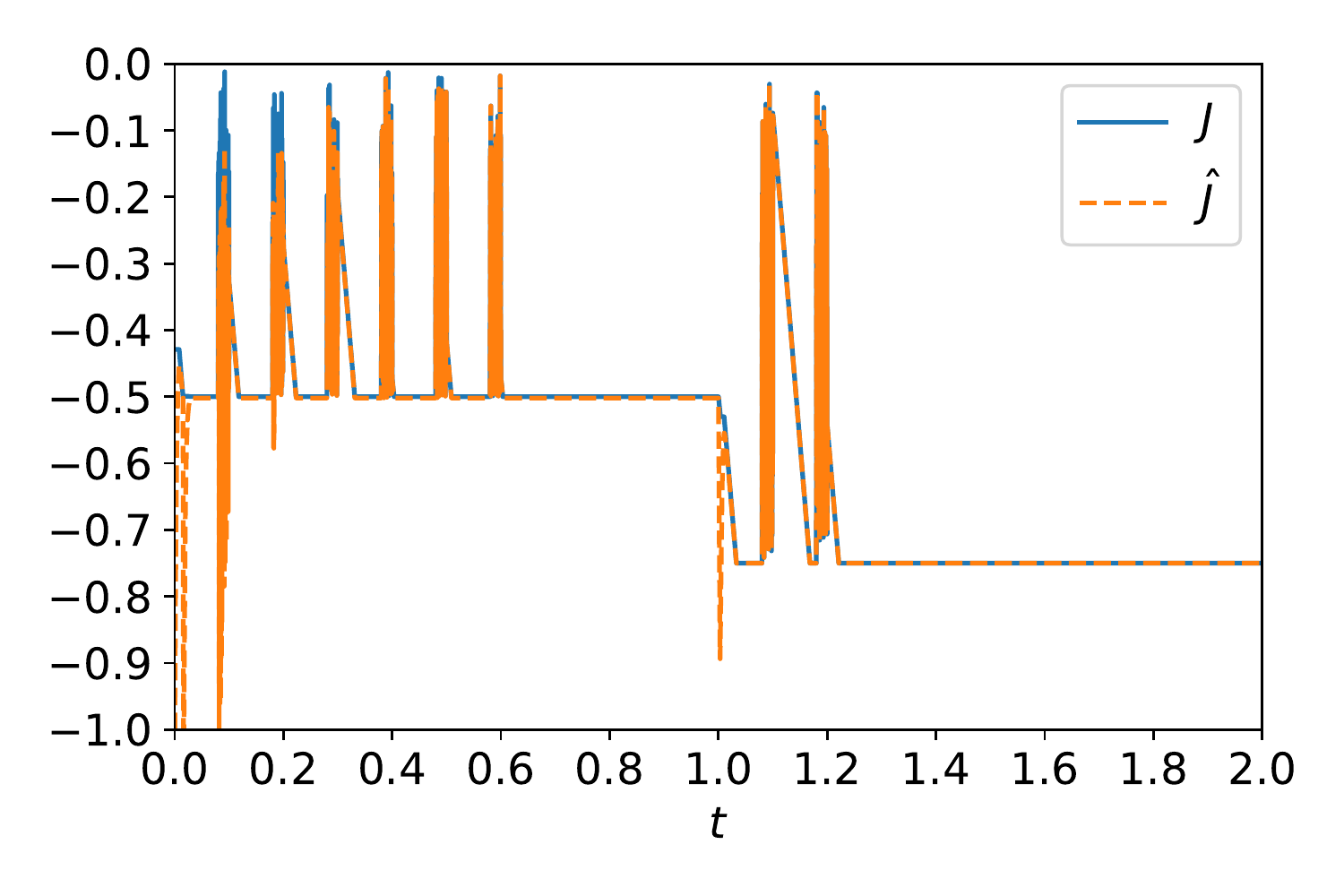}\label{fig:eg-2-rtg-cost-pe}}%
\hfill%
\subfloat[Attacker's cost]{\includegraphics[width=.49\columnwidth,max width=128pt,trim=1.5em 4ex 1.5em 3.5ex,clip]{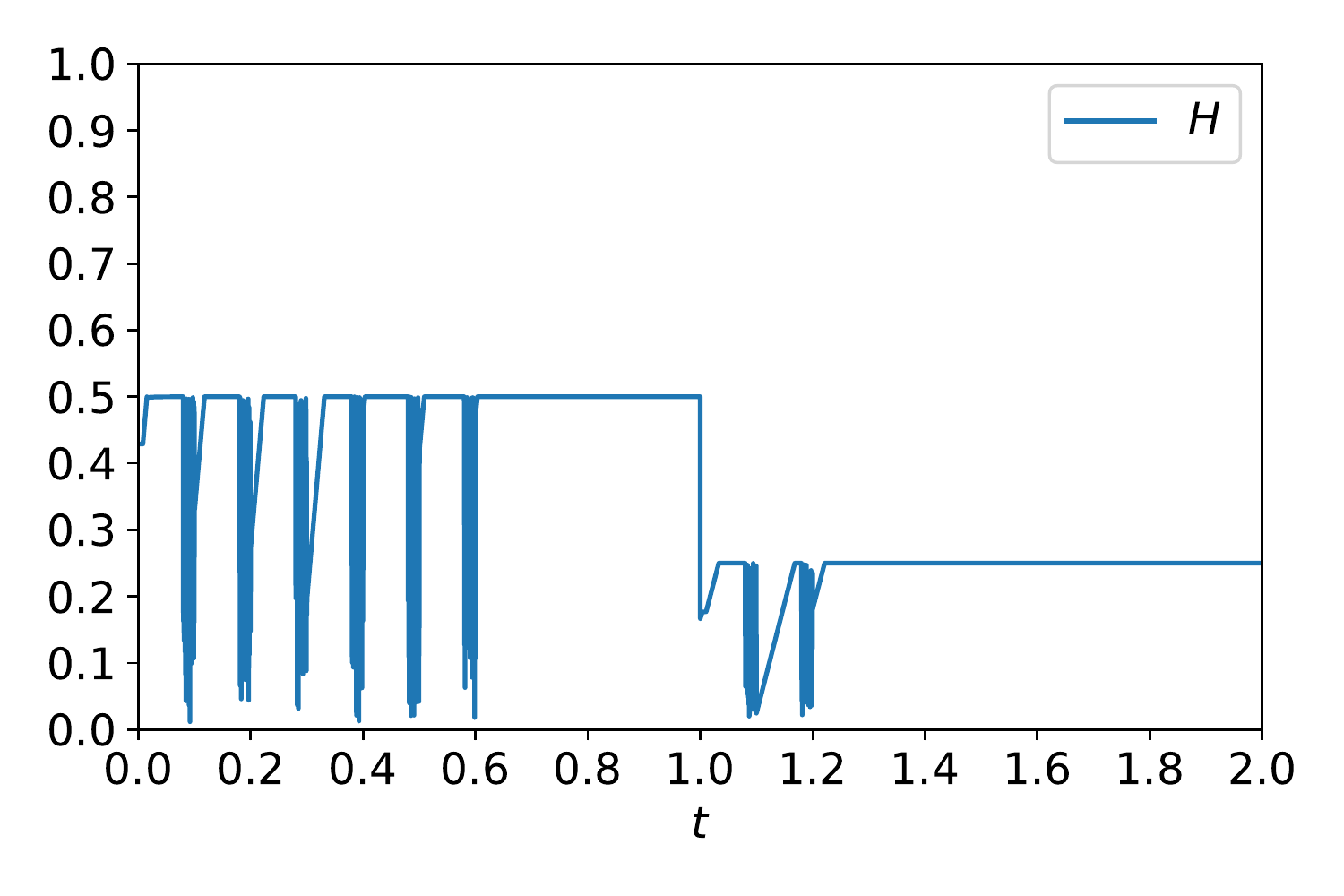}\label{fig:eg-2-atk-cost-pe}}%
\caption{Simulation results for $ L = 2 $ w/ PE (horizontal axis: $ \times 10^4 $ units of time). In the first $ 10^4 $ units of time, the observation error converges to $ e_\obs = 0 $, the router's action $ r $ converges to the Stackelberg equilibrium action $ r^* = (0.5, 0.5) $, the router's and attacker's costs converge to $ J = \hat J = -H = -0.5 $; in the second $ 10^4 $ units of time, the attacker switches to a new cost function, the router's action $ r $ converges to an $ \varepsilon $ Stackelberg action near $ \bar r^* = (0.25, 0.75) $, the router's actual and predicted costs converge to $ J = \hat J = -0.75 $, while the attacker's cost converges to $ H = 0.25 $.}\label{fig:eg-2-pe}%
\end{figure}

\begin{figure}[!htbp]
\centering
\subfloat[Router's actual and predicted cost functions at $ T = 10^4 $]{\includegraphics[width=.49\columnwidth,max width=128pt,trim=1.5em 3.5ex 1.5em 3.5ex,clip]{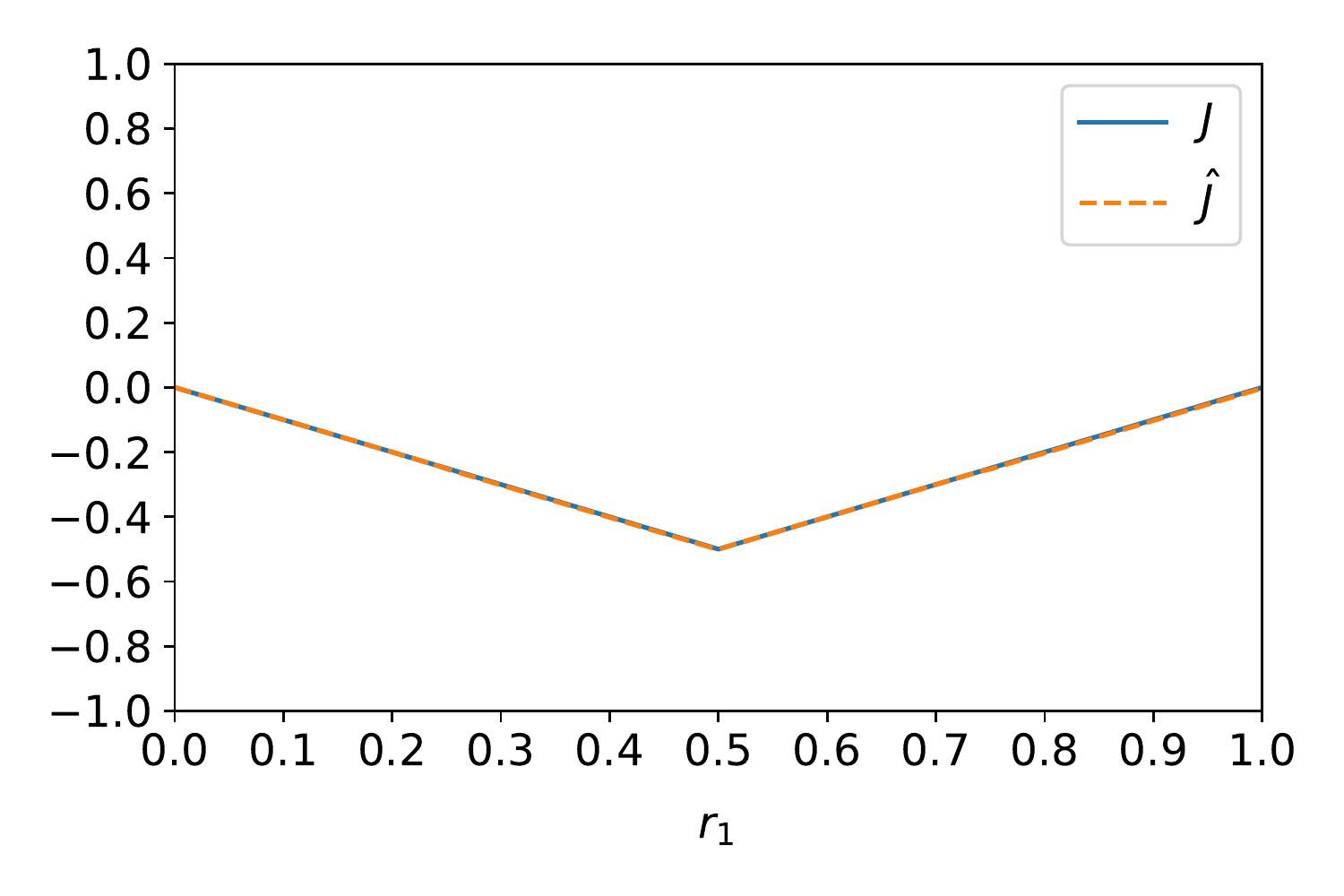}\label{fig:eg-2-fcn-1-pe}}%
\hfill%
\subfloat[Router's actual and predicted cost functions at $ T = 2 \times 10^4 $]{\includegraphics[width=.49\columnwidth,max width=128pt,trim=1.5em 3.5ex 1.5em 3.5ex,clip]{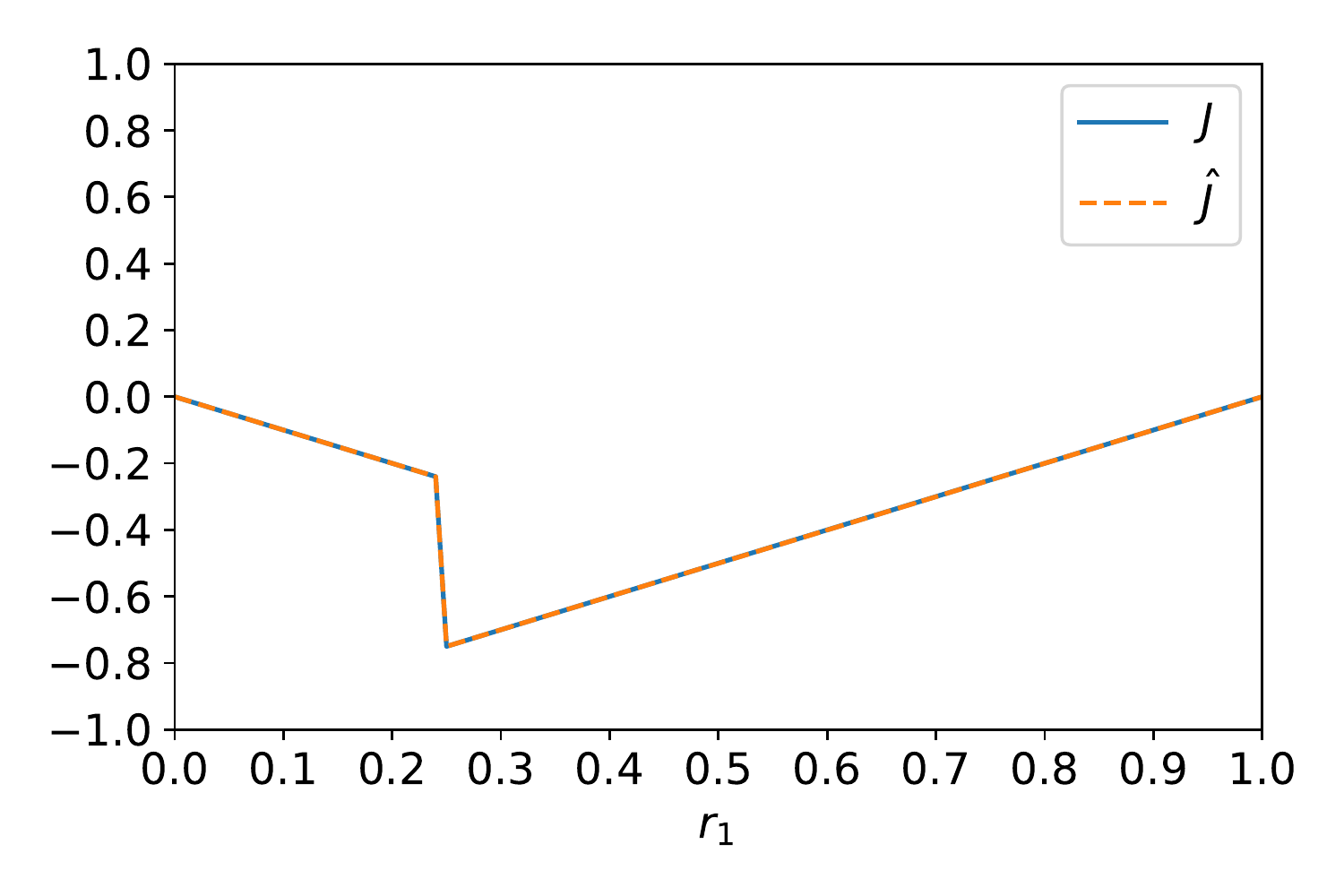}\label{fig:eg-2-fcn-2-pe}}%
\caption{Router's actual cost function $ r_1 \mapsto J(r, f(r)) $ and predicted cost function $ r_1 \mapsto \hat J(\hat\theta(T), r) $ for $ L = 2 $ w/ PE. In both scenarios, the predicted cost function is accurate everywhere.}\label{fig:eg-2-fcn-pe}%
\end{figure}

For the case without enforcing PE in Fig.~\ref{fig:eg-2}, in the first $ 10^4 $ units of time, the observation error $ e_\obs $ converges to $ 0 $, and the router's action $ r $ converges to the Stackelberg equilibrium action $ r^* = (0.5, 0.5) $, despite that the parameter estimate $ \hat\theta $ does not converge to the actual value $ \theta $, as shown in Fig.~\subref*{fig:eg-2-fcn-1}. In Fig.~\ref{fig:eg-2-pe}, we enforce PE by adding some random noise to $ r $ for a short period of time while the observation error remains small. In this case, in the first $ 10^4 $ units of time, the observation error $ e_\obs $ converges to $ 0 $, the router's action $ r $ converges to the Stackelberg equilibrium action $ r^* $, and the parameter estimate $ \hat\theta $ converges to the actual value $ \theta $, as shown in Fig.~\subref*{fig:eg-2-fcn-1-pe}. In both cases with and without PE, we also simulate the scenario where, after $ 10^4 $ units of time, the attacker starts focusing more on disrupting link $ 1 $ by switching to a new cost function defined by
\begin{equation*}
    \bar H(a, r) := u_1 + u_2/3.
\end{equation*}
The results in \cite[Cor.~2]{YangHespanha2021} allow us to conclude that, for the non-zero-sum game defined by $ (\cR, \cA, J, \bar H) $, the attacker's best response to a router's action $ r $ is to set $ a_l = 1 $ on the link $ l \in \{1,\, 2\} $ that corresponds to the larger one in $ \{r_1,\, r_2/3\} $. The corresponding actual value $ \theta $ in \eqref{eq:match}, written as the tensor form in \eqref{eq:sim-rbf}, becomes
\begin{equation*}
    \theta = \begin{bmatrix}
        0 & 1& 1 & 1 \\
        1 & 0 & 0 & 0
    \end{bmatrix}.
\end{equation*}
As established in \cite[Th.~4 and~Cor.~5]{YangHespanha2021}, there is no Stackelberg equilibrium action as defined by Definition~\ref{dfn:stbg} for $ (\cR, \cA, J, \bar H) $; however, there are $ \varepsilon $ Stackelberg actions near $ \bar r^* = (0.25, 0.75) $ for sufficiently small $ \varepsilon > 0 $. The corresponding simulation results in Fig.~\ref{fig:eg-2},~\subref*{fig:eg-2-fcn-2},~\ref{fig:eg-2-pe}, and~\subref*{fig:eg-2-fcn-2-pe} show that our adaptive learning approach is able to identify this switch in the attack, as the router's action $ r $ converges to $ \varepsilon $ Stackelberg actions near $ \bar r^* $ in both Fig.~\ref{fig:eg-2} and~\ref{fig:eg-2-pe}, and, when PE is enforced, the parameter estimate $ \hat\theta $ converges to the new actual value $ \theta $, as shown in Fig.~\subref*{fig:eg-2-fcn-2-pe}.

\subsection{Network with three parallel links}
Consider a network with $ L = 3 $ parallel links, capacity $ c_0 = 1 $, total desired legitimated traffic $ R = L c_0/2 = 1.5 $, and attack budget $ A = \lceil L c_0/2 \rceil = 2 $. We set the constant $ n_\rbf = 20 $. Then $ n_\theta = 1200 $ and the parameter set $ \Theta = [0, 1]^{1200} $. Following \cite[Cor.~2]{YangHespanha2021}, the attacker's best response to a router action is to set $ a_{l_1} = a_{l_2} = 1 $ on the two links $ l_1, l_2 \in \{1,\, 2,\, 3\} $ with largest $ r_l $. However, there is no $ \theta \in \Theta $ such that \eqref{eq:match} holds for attacker's strategy $ f $ and the function $ \hat f $ defined by \eqref{eq:sim-rbf}. To see the mismatch between $ f(r) $ and $ \hat f(\theta, r) $, one could consider the router's actions
\begin{equation*}
    r \in \{(0.45 + \varepsilon, 0.5 - \varepsilon, 0.55): 0 < \varepsilon < 0.05\}.
\end{equation*}
Comparing to \eqref{eq:sim-rbf}, we see that $ (r_1, r_2) $ belongs to the hypercube with the center $ (0.475, 0.475) $. However, the attacker's best response satisfies $ a_1 = 1 $ if $ \varepsilon > 0.025 $ and $ a_1 = 0 $ if $ \varepsilon < 0.025 $. Hence the corresponding scalar component of $ \theta $ will always yield an error of at least $ 0.5 $, and thus the mismatch threshold $ \varepsilon_f $ in \eqref{eq:mismatch-bnd} or \eqref{eq:mismatch-bnd-new} is at least $ 0.5 A/c_0 = 1 $. However, in the simulations below we are able to obtain much smaller observation errors near the Stackelberg equilibrium action $ r^* = (0.5, 0.5, 0.5) $ established in \cite[Cor.~5]{YangHespanha2021}. Due to the mismatch and motivated by Remark~\ref{rmk:mismatch}, we focus our attention on the performance of the estimation and learning algorithms \eqref{eq:est-dyn} and \eqref{eq:ctrl-dyn}.

\begin{figure}[!htbp]
\centering
\subfloat[Observation error]{\includegraphics[width=.49\columnwidth,max width=128pt,trim=1.5em 4ex 1.5em 3.5ex,clip]{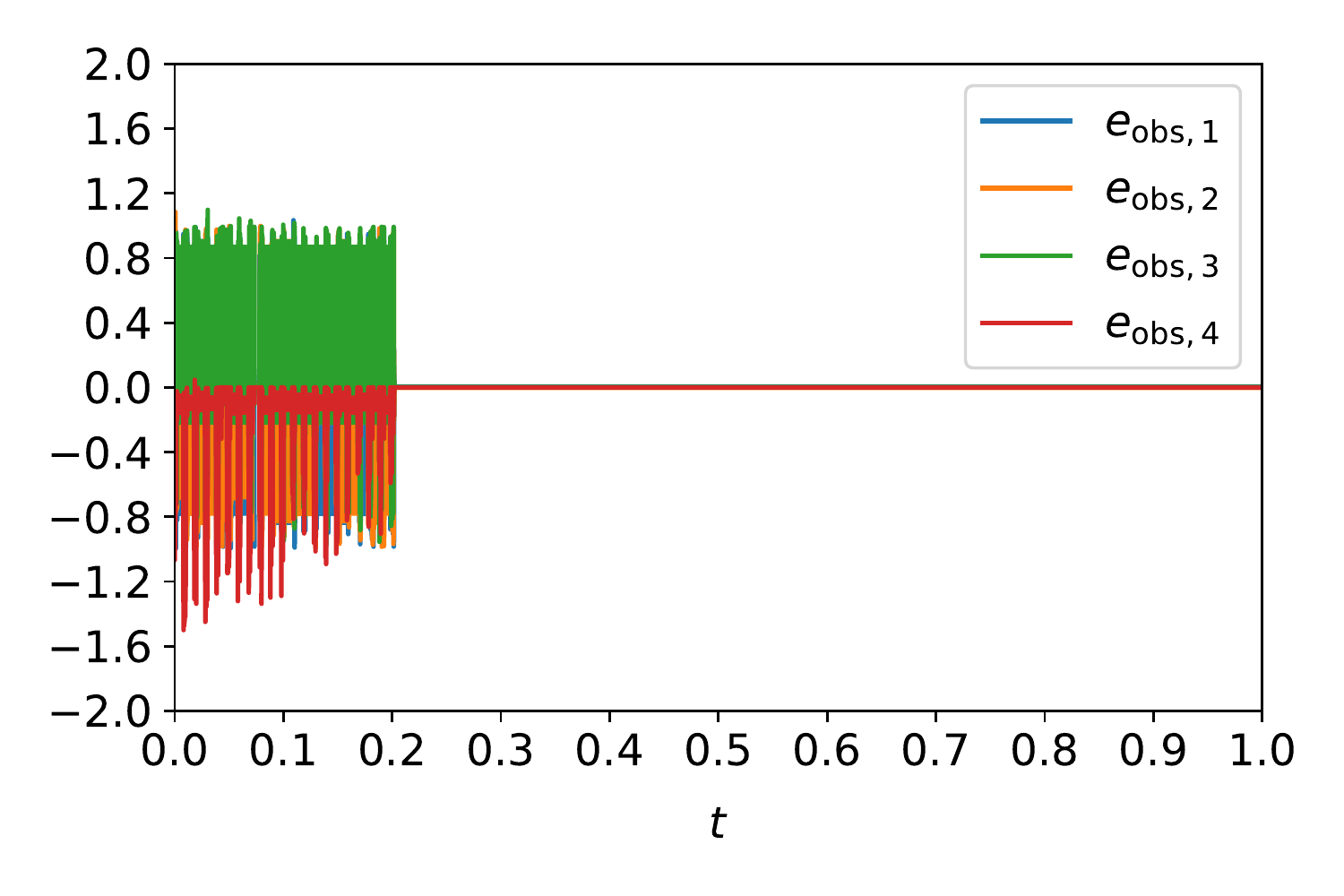}\label{fig:eg-3-obs}}%
\hfill%
\subfloat[Router's action]{\includegraphics[width=.49\columnwidth,max width=128pt,trim=1.5em 4ex 1.5em 3.5ex,clip]{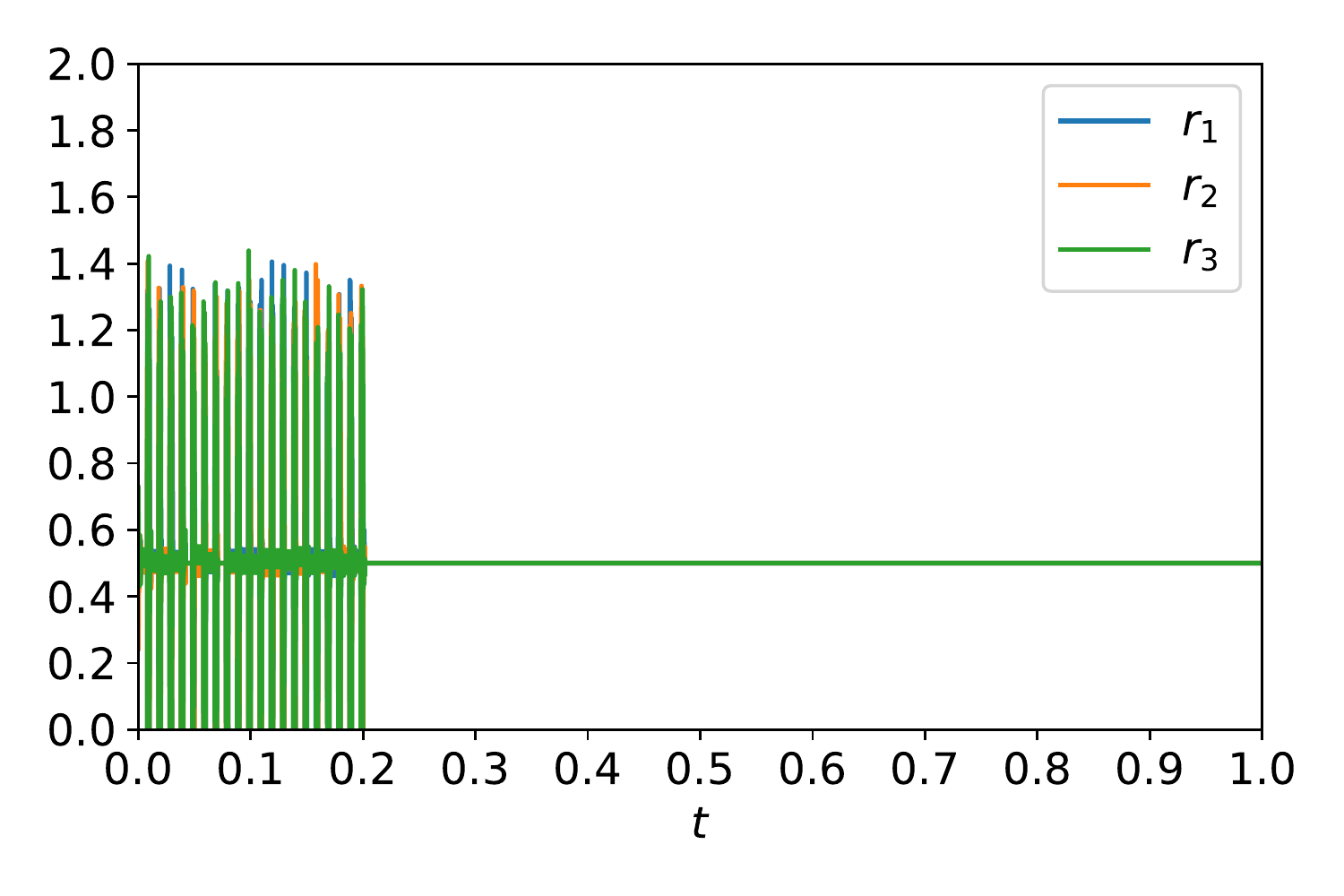}\label{fig:eg-3-rtg}}%
\\%
\subfloat[Router's actual and predicted costs]{\includegraphics[width=.49\columnwidth,max width=128pt,trim=1.5em 4ex 1.5em 3.5ex,clip]{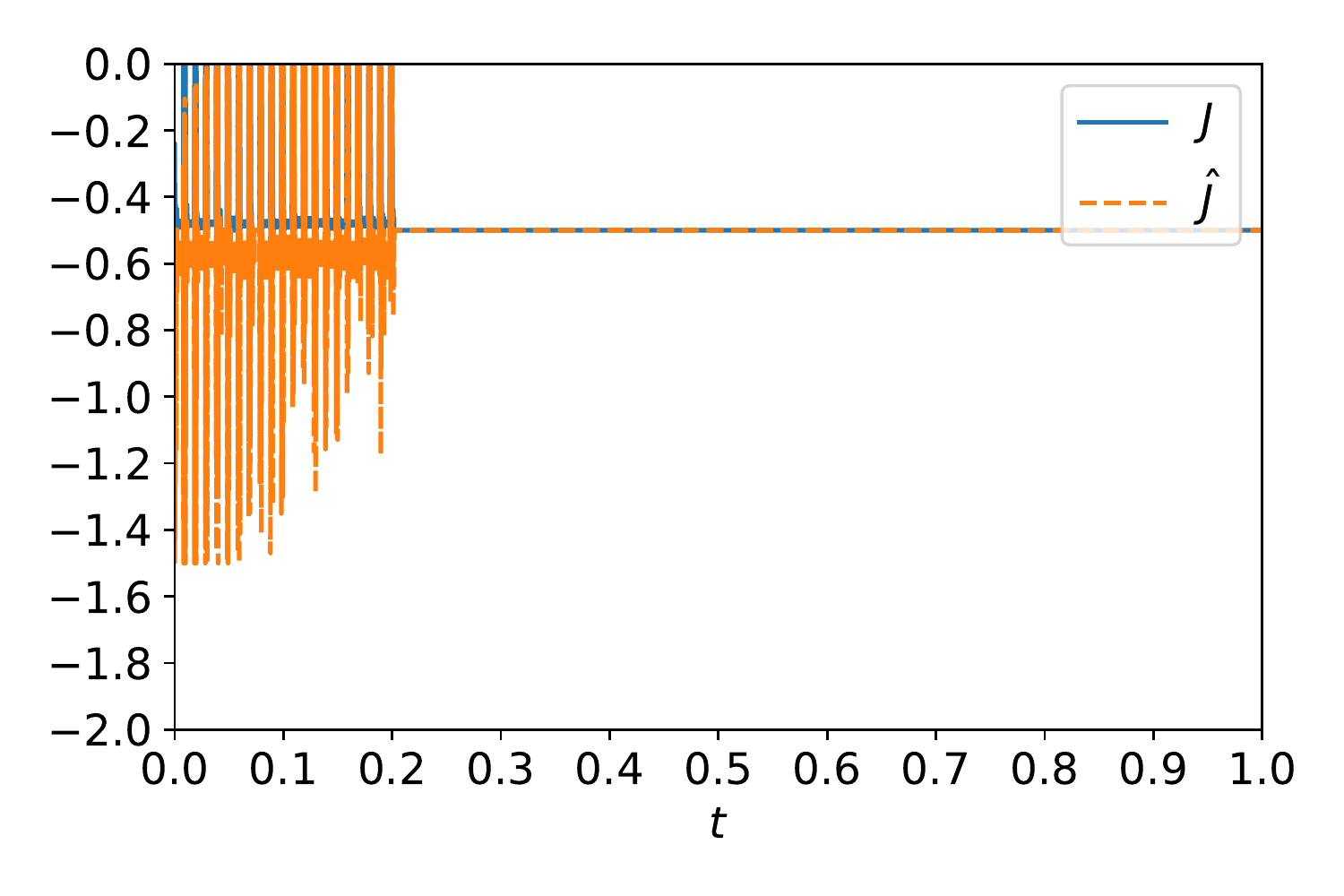}\label{fig:eg-3-rtg-cost}}%
\hfill%
\subfloat[Attacker's cost]{\includegraphics[width=.49\columnwidth,max width=128pt,trim=1.5em 4ex 1.5em 3.5ex,clip]{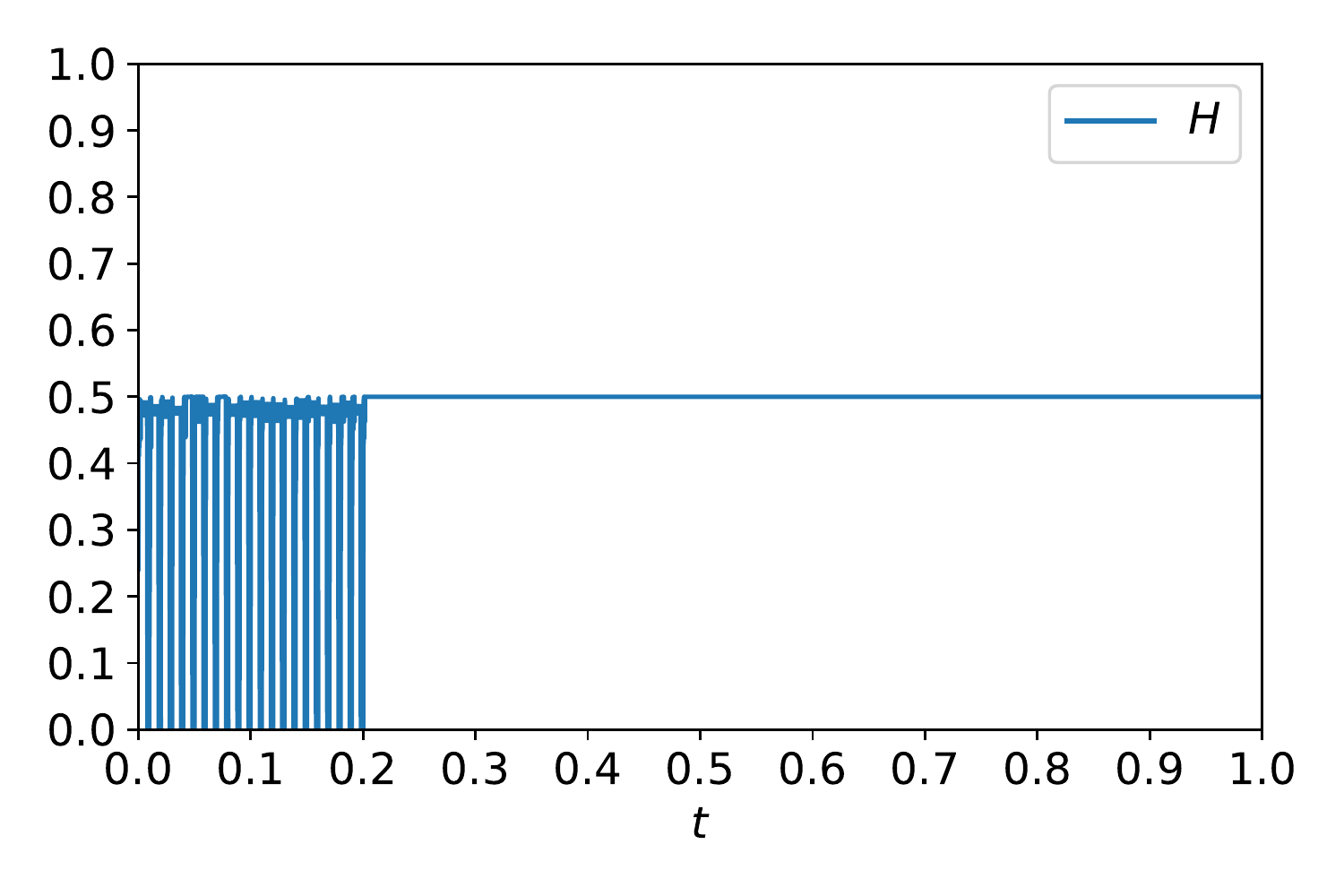}\label{fig:eg-3-atk-cost}}%
\caption{Simulation results for $ L = 3 $ and the cost function $ H $ w/ PE (horizontal axis: $ \times 10^5 $ units of time). The observation error converges to $ e_\obs = 0 $, the router's action $ r $ converges to the Stackelberg equilibrium action $ r^* = (0.5, 0.5, 0.5) $, the router's and attacker's costs converge to $ J = \hat J = -H = -0.5 $.}\label{fig:eg-3}%
\end{figure}

\begin{figure}[!htbp]
\centering
\subfloat[Observation error]{\includegraphics[width=.49\columnwidth,max width=128pt,trim=1.5em 4ex 1.5em 3.5ex,clip]{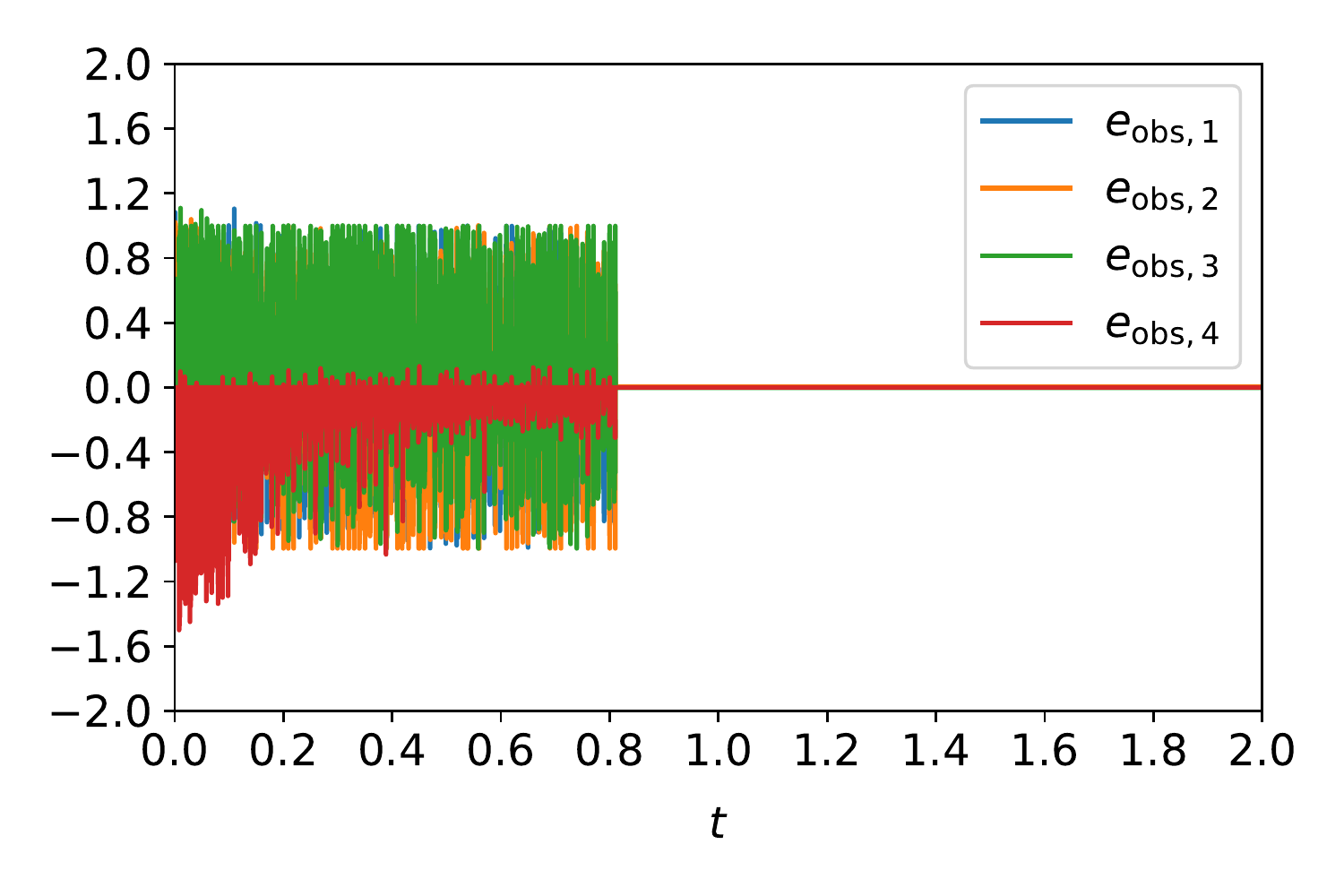}\label{fig:eg-3-obs-new}}%
\hfill%
\subfloat[Router's action]{\includegraphics[width=.49\columnwidth,max width=128pt,trim=1.5em 4ex 1.5em 3.5ex,clip]{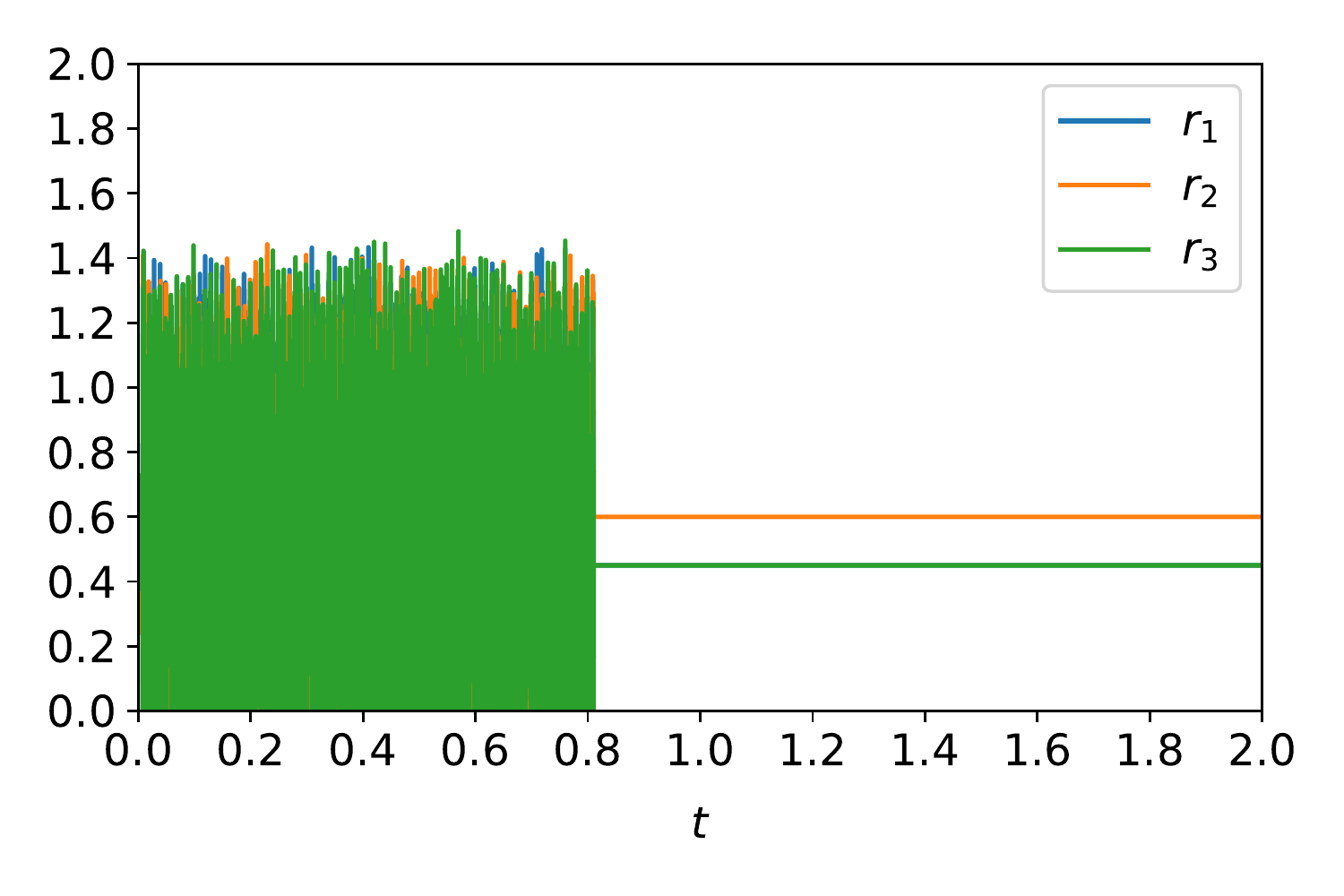}\label{fig:eg-3-rtg-new}}%
\\%
\subfloat[Router's actual and predicted costs]{\includegraphics[width=.49\columnwidth,max width=128pt,trim=1.5em 4ex 1.5em 3.5ex,clip]{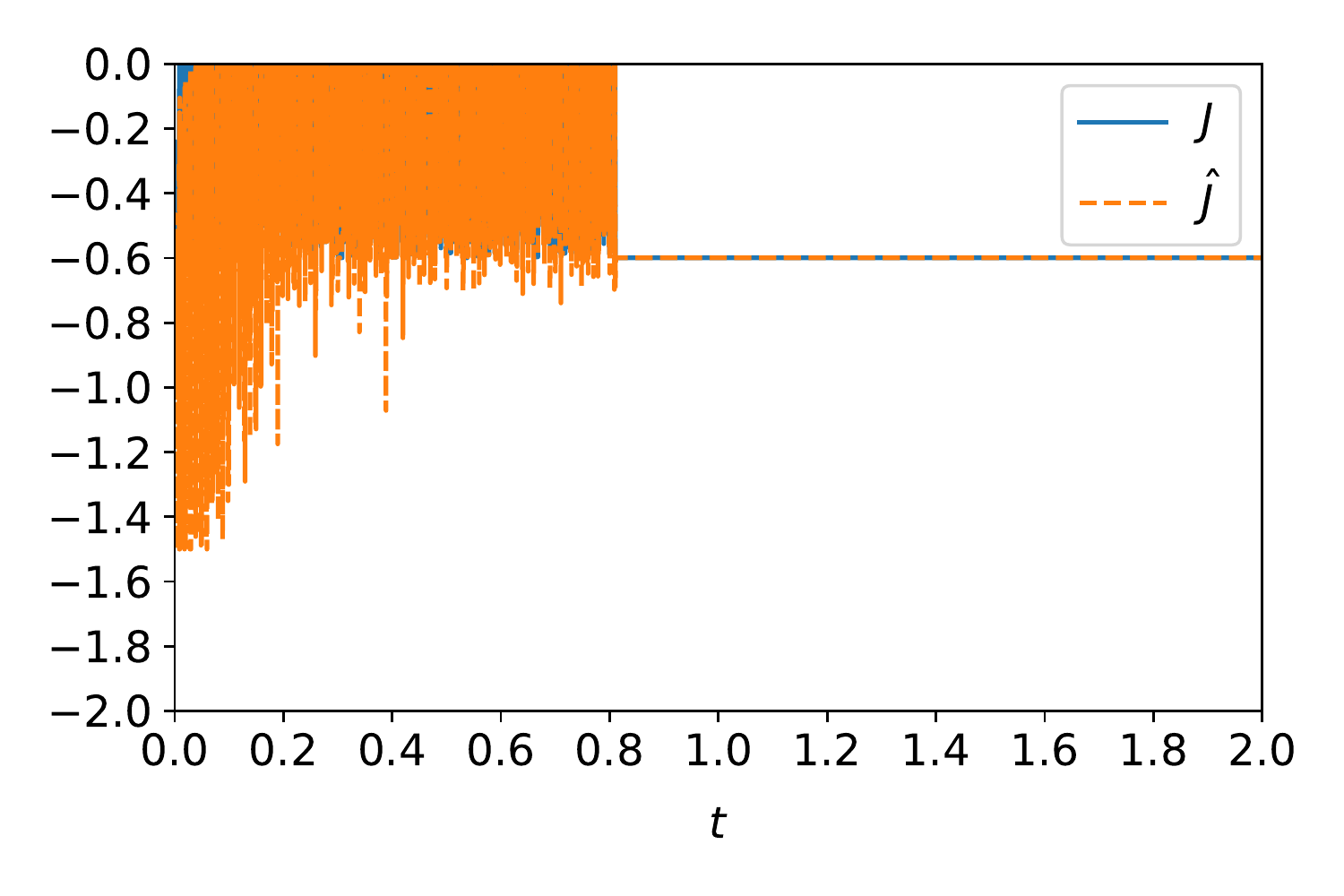}\label{fig:eg-3-rtg-cost-new}}%
\hfill%
\subfloat[Attacker's cost]{\includegraphics[width=.49\columnwidth,max width=128pt,trim=1.5em 4ex 1.5em 3.5ex,clip]{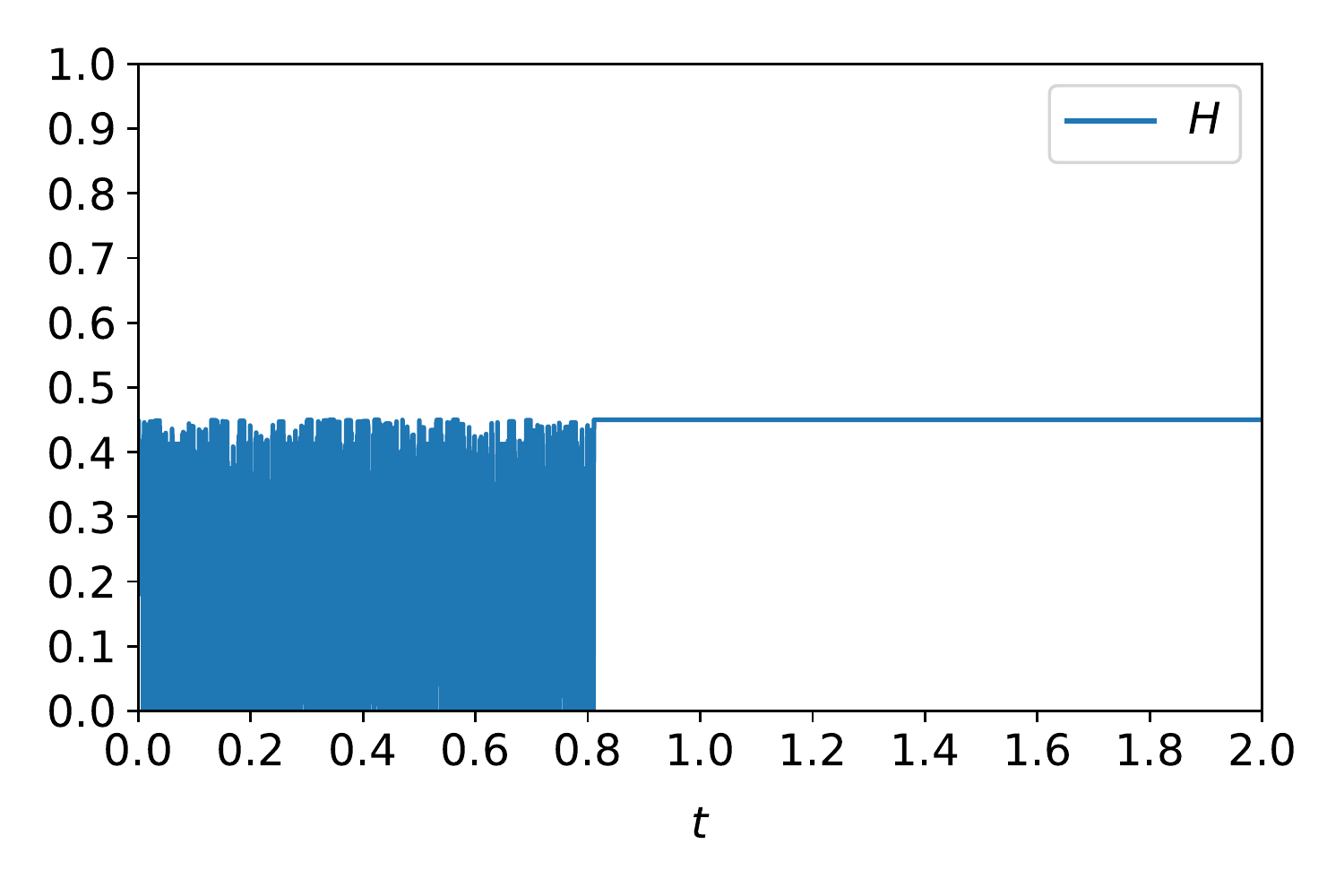}\label{fig:eg-3-atk-cost-new}}%
\caption{Simulation results for $ L = 3 $ and the cost function $ \bar H $ w/ PE (horizontal axis: $ \times 10^5 $ units of time). The observation error $ e_\obs $ converges to $ e_\obs = 0 $, the router's action $ r $ converges to an $ \varepsilon $ Stackelberg action near $ \bar r^* = (0.45, 0.6, 0.45) $, the router's actual and predicted costs converge to $ J = \hat J = -0.6 $, while the attacker's cost converges to $ H = 0.45 $.}\label{fig:eg-3-new}%
\end{figure}

In Fig.~\ref{fig:eg-3}, the observation error $ e_\obs $ converges to $ 0 $, and the router's action $ r $ converges to the Stackelberg equilibrium action $ r^* = (0.5, 0.5, 0.5) $. In Fig.~\ref{fig:eg-3-new}, the attacker starts focusing less on disrupting links $ 2 $ by switching to the new cost function defined by
\begin{equation*}
    \bar H(a, r) := u_1 + 3 u_2/4 + u_3.
\end{equation*}
Following \cite[Cor.~2]{YangHespanha2021}, for the non-zero-sum game defined by $ (\cR, \cA, J, \bar H) $, the attacker's best response to a router's action $ r $ is to set $ a_{l_1} = a_{l_2} = 1 $ on the two links $ l_1, l_2 \in \{1,\, 2,\, 3\} $ that correspond to the largest two in $ \{r_1,\, 3 r_2/4,\, r_3\} $. As established in \cite[Th.~4 and~Cor.~5]{YangHespanha2021}, there is no Stackelberg equilibrium action as defined by Definition~\ref{dfn:stbg} for $ (\cR, \cA, J, \bar H) $, but there are $ \varepsilon $ Stackelberg actions near $ \bar r^* = (0.45, 0.6, 0.45) $ for sufficiently small $ \varepsilon > 0 $. The corresponding simulation results are plotted in Fig.~\ref{fig:eg-3-new}, where the observation error $ e_\obs $ converges to $ 0 $, and the router's action $ r $ converges to an $ \varepsilon $ Stackelberg action near $ \bar r^* $.

\section{Conclusion}\label{sec:end}
This paper considered a two-player Stackelberg game where the leader only had partial knowledge of the follower's action set and cost function, and had to estimate the follower's strategy using a family of parameterized functions. We designed an adaptive learning approach that simultaneously estimated the follower's strategy based on past observations and minimized the leader's cost predicted using the latest estimation. Our approach was proved to guarantee that leader's actual and predicted costs becomes essentially indistinguishable in finite time, and the first-order necessary condition for optimality holds asymptotically for the predicted cost. Moreover, we provided a PE condition for ensuring estimation accuracy and studied the case with model mismatch. The results were illustrated via a simulation example motivated by DDoS attacks.

A feature of our estimation algorithm \eqref{eq:est-dyn} is that the norm of the estimation error $ \theta - \theta^* $ is monotonically decreasing and the observation error $ e_\obs $ becomes bounded in norm by the preselected, arbitrarily small threshold $ \varepsilon_\obs $ in finite time, regardless of how the leader's action $ r $ is adjusted. A future research direction is to extend our adaptive learning approach by adopting more efficient estimation and optimization algorithms for more complex applications. Some preliminary results on employing a neural network to predict the follower's response can be found in \cite{YangHespanha2021}. Other future research topics include to relax the affine condition in Assumption~\ref{ass:reg}, and to extend the current results to Stackelberg games on distributed networks.

\appendices
\section{Projected dynamical systems}\label{apx:proj-dyn}
Let $ \cS \subset \R^n $ be a compact convex set and $ g: \cS \to \R^n $ be a continuous function. In this section, we provide some preliminary results on the existence and convergence of solutions for the \emph{projected dynamical system}
\begin{equation}\label{eq:proj-dyn}
    \dot x = [g(x)]_{T_\cS(x)}.
\end{equation}
The difficulty in analyzing \eqref{eq:proj-dyn} lies in the fact that its right-hand side is only defined in the domain $ \cS $ and is potentially discontinuous in the boundary $ \partial\cS $ due to the projection $ [\cdot]_{T_\cS(x)} $. Therefore, we consider the concept of viable Carath\'{e}odory solution, that is, a solution to \eqref{eq:proj-dyn} on an interval $ L \subset \R_+ $ is an absolutely continuous function $ x: L \to \cS $ such that \eqref{eq:proj-dyn} holds almost everywhere on $ L $. In particular, it requires $ x(t) \in \cS $ for all $ t \in L $. The following result establishes the existence of solutions for the projected dynamical system \eqref{eq:proj-dyn}.
\begin{lem}\label{lem:proj-dyn}
For each $ x_0 \in \cS $, there exists a solution $ x $ to \eqref{eq:proj-dyn} on $ \R_+ $ with $ x(0) = x_0 $.
\end{lem}
\begin{proof}
As the function $ g $ is continuous on the compact set $ \cS $, it is upper bounded in norm and thus a Marchaud map \cite[Def.~2.2.4, p.~62]{Aubin1991}. Then Lemma~\ref{lem:proj-dyn} follows from \cite[Th.~10.1.1, p.~354]{Aubin1991}.
\end{proof}

In the following, we establish an invariance principle for the case where $ g $ is defined by a gradient descent process.
\begin{prop}\label{prop:proj-dyn-lasalle}
Suppose that the function $ g $ in \eqref{eq:proj-dyn} satisfies
\begin{equation}\label{eq:proj-dyn-grad}
    g(z) = -\nabla V(z)^\top \qquad \forall z \in \cS
\end{equation}
for some function $ V: \cS \to \R $. Then every solution $ x $ to \eqref{eq:proj-dyn} satisfies
\begin{equation}\label{eq:proj-dyn-lasalle}
    \lim_{t \to \infty} [g(x(t))]_{T_\cS(x(t))} = 0.
\end{equation}
\end{prop}

To prove Proposition~\ref{prop:proj-dyn-lasalle}, we extend the projected differential equation \eqref{eq:proj-dyn} to the differential inclusion
\begin{equation}\label{eq:proj-dyn-incl}
    \dot x \in G(x)
\end{equation}
with the set-valued function $ G: \cS \rightrightarrows \R^n $ defined by\footnote{The extension from the projected differential equation \eqref{eq:proj-dyn} to the differential inclusion \eqref{eq:proj-dyn-incl} is inspired by similar extensions in \cite{Henry1973} and \cite[p.~354]{Aubin1991}, and is specifically designed to simplified the proof of Proposition~\ref{prop:proj-dyn-lasalle}.}
\begin{equation}\label{eq:proj-dyn-incl-fcn}
\begin{aligned}
    G(z) &:= \{g(z) - w: w \in N_\cS(z)\} \\
    &\;\quad\,\cap \big\| g(z) - [g(z)]_{T_\cS(z)} \big\| \bB(g(z)).
\end{aligned}
\end{equation}
As $ [g(z)]_{T_\cS(z)} \in G(z) $ for all $ z \in \cS $, a solution to \eqref{eq:proj-dyn} is also a solution to \eqref{eq:proj-dyn-incl}. Hence we prove Proposition~\ref{prop:proj-dyn-lasalle} by applying an invariance theorem for differential inclusions to \eqref{eq:proj-dyn-incl}, which requires the following continuity property.
\begin{lem}
The set-valued function $ G $ defined by \eqref{eq:proj-dyn-incl-fcn} is upper semicontinuous on $ \cS $.
\end{lem}
\begin{proof}
As $ \cS $ is a compact and convex, the set-valued map $ z \mapsto T_\cS(z) $ is lower semicontinuous on $ \cS $ \cite[Th.~5.1.7, p.~162]{Aubin1991}. Also, as $ g $ is continuous on $ \cS $, the map $ (z, w) \mapsto -\|g(z) - w\| $ is continuous on $ \cS \times \R^n $. Hence the map $ z \mapsto \sup_{w \in T_\cS(z)} -\|g(z) - w\| = -\inf_{w \in T_\cS(z)} \|g(z) - w\| = -\big\| g(z) - [g(z)]_{T_\cS(z)} \big\| $ is lower semicontinuous on $ \cS $ \cite[Th.~2.1.6, p.~59]{Aubin1991}, that is, the map $ z \mapsto \big\| g(z) - [g(z)]_{T_\cS(z)} \big\| $ is upper semicontinuous on $ \cS $. Moreover, the map $ z \mapsto \{g(z) - w: w \in N_\cS(z)\} $ is closed. Hence $ G $ is upper semicontinuous on $ \cS $ \cite[Cor.~2.2.3, p.~61]{Aubin1991}.
\end{proof}

\begin{proof}[Proof of Proposition~\ref{prop:proj-dyn-lasalle}]
Suppose that
\begin{equation}\label{eq:proj-dyn-incl-lya}
    g(z)^\top w \geq \big\| [g(z)]_{T_\cS(z)} \big\|^2 \qquad \forall z \in \cS, \forall w \in G(z).
\end{equation}
Then the function $ V $ satisfies
\begin{equation*}
    \nabla V(z)\, w \leq -\big\| [g(z)]_{T_\cS(z)} \big\|^2 \qquad \forall z \in \cS, \forall w \in G(z).
\end{equation*}
Note that the set-valued function $ G $ is upper semicontinuous, and the set $ G(z) $ is nonempty, compact, and convex for every $ z \in \cS $. Hence the invariance theorem \cite[Th.~2.11]{Ryan1998} implies that every solution to \eqref{eq:proj-dyn-incl}, and therefore every solution $ x $ to \eqref{eq:proj-dyn}, approaches the largest invariant set in $ \{z \in \cS: \|[g(z)]_{T_\cS(z)}\| = 0\} $. Hence \eqref{eq:proj-dyn-lasalle} holds as $ g $ is continuous on the compact set $ \cS $.

It remains to show that \eqref{eq:proj-dyn-incl-lya} holds. Consider arbitrary $ z \in \cS $ and $ w \in G(z) $. First, as $ g(z) - w \in N_\cS(z) $, from \eqref{eq:normal-cone} and \eqref{eq:tan-proj-cone} we have $ (g(z) - w)^\top [g(z)]_{T_\cS(z)} \leq 0 $, and thus
\begin{equation*}
\begin{aligned}
	&\quad\,\, \|w\|^2 - \big\| [g(z)]_{T_\cS(z)} \big\|^2 \\
	&\geq \|w\|^2 - \big\| [g(z)]_{T_\cS(z)} \big\|^2 - 2 (g(z) - w)^\top [g(z)]_{T_\cS(z)} \\
	&= \big\| w - [g(z)]_{T_\cS(z)} \big\|^2 \geq 0,
\end{aligned}
\end{equation*}
where the equality follows partially from \eqref{eq:tan-proj-zero}. Second, as $ w \in G(z) \subset \big\| g(z) - [g(z)]_{T_\cS(z)} \big\| \bB(g(z)) $, we have $ \|g(z) - w\| \leq \big\| g(z) - [g(z)]_{T_\cS(z)} \big\| $, and thus
\begin{equation*}
\begin{aligned}
	2 g(z)^\top w &= \|w\|^2 + \|g(z)\|^2 - \|g(z) - w\|^2 \\
	&\geq \|w\|^2 + \|g(z)\|^2 - \big\| g(z) - [g(z)]_{T_\cS(z)} \big\|^2 \\
	&= \|w\|^2 + \big\| [g(z)]_{T_\cS(z)} \big\|^2 \geq 2 \big\| [g(z)]_{T_\cS(z)} \big\|^2,
\end{aligned}
\end{equation*}
where the second equality follows partially from \eqref{eq:tan-proj-zero}. Hence \eqref{eq:proj-dyn-incl-lya} holds.
\end{proof}

\section{Proof of technical lemmas}
\subsection{Proof of Lemma~\ref{lem:obs-est}}\label{apx:obs-est}
As the map $ \hat\theta \mapsto \hat f(\hat\theta, r) $ is affine, its Jacobian matrix $ \jacob{\theta} \hat f(r) $ is independent of $ \hat\theta $. Thus for a fixed $ r \in \cR $, we have
\begin{equation*}
    \hat f(\hat\theta, r) - a = \hat f(\hat\theta, r) - \hat f(\theta, r) = \jacob{\theta} \hat f(t) (\hat\theta - \theta).
\end{equation*}
Next, consider the function $ g: [0, 1] \to \R $ defined by
\begin{equation*}
    g(\rho) := J \big( r, \rho \hat f(\hat\theta, r) + (1 - \rho)\, a \big),
\end{equation*}
which is continuously differentiable as $ J $ is continuously differentiable. Then
\begin{equation*}
    \hat J(r, \hat\theta) - J(r, a) = g(1) - g(0) = \int_{0}^{1} \frac{\d g(\rho)}{\d\rho} \d\rho,
\end{equation*}
in which
\begin{equation*}
\begin{aligned}
    \frac{\d g(\rho)}{\d\rho} &= \jacob{a} J \big( r, \rho \hat f(\hat\theta, r) + (1 - \rho)\, a \big) \big( \hat f(\hat\theta, r) - a \big) \\
    &= \jacob{a} J \big( r, \rho \hat f(\hat\theta, r) + (1 - \rho)\, a \big) \jacob{\theta} \hat f(r) (\hat\theta - \theta).
\end{aligned}
\end{equation*}
Hence
\begin{equation*}
\begin{aligned}
    &\quad\,\, \hat J(r, \hat\theta) - J(r, a) \\
    &= \int_{0}^{1} \jacob{a} J(r, \rho \hat f(\hat\theta, r) + (1 - \rho) a) \jacob{\theta} \hat f(r) (\hat\theta - \theta) \d\rho \\
    &= \bigg( \int_{0}^{1} \jacob{a} J(r, \rho \hat f(\hat\theta, r) + (1 - \rho) a) \d\rho \bigg) \jacob{\theta} \hat f(r) (\hat\theta - \theta).
\end{aligned}
\end{equation*}

\subsection{Proof of Lemma~\ref{lem:dyn-soln}}\label{apx:dyn-soln}
Consider an arbitrary $ (\hat\theta_0, r_0) \in \Theta \times \cR $, and let $ \lambda_e(0) $ be the corresponding value given by \eqref{eq:obs-err} and \eqref{eq:est-sw}. Suppose $ \lambda_e(0) = \lambda_\theta $, that is, $ \|e_\obs(0)\| > \varepsilon_\obs' $. In the following, we construct a solution to \eqref{eq:est-dyn} and \eqref{eq:ctrl-dyn} on $ \R_+ $ with $ (\hat\theta(0), r(0)) = (\hat\theta_0, r_0) $ recursively. The case where $ \lambda_e(0) = 0 $ can be proved based on the same construction starting with the second step.

First, consider the system defined by \eqref{eq:est-dyn} and \eqref{eq:ctrl-dyn} with $ \lambda_e \equiv \lambda_\theta $, that is,
\begin{equation}\label{eq:dyn-1}
\begin{aligned}
    \dot{\hat\theta} &= \big[ {-\lambda_\theta} K(r, f(r), \hat\theta)^\top K(r, f(r), \hat\theta) (\hat\theta - \theta) \big]_{T_\Theta(\hat\theta)}, \\
    \dot r &= \big[ {-\lambda_r} \jacob{r} \hat J(r, \hat\theta)^\top \big]_{T_\cR(r)},
\end{aligned}
\end{equation}
which can be modeled using the projected dynamical system \eqref{eq:proj-dyn} in Appendix~\ref{apx:proj-dyn} with the state $ x := (\hat\theta, r) $ and the set $ \cS := \Theta \times \cR $. In particular, $ f $ is continuous due to \eqref{eq:match} and Assumption~\ref{ass:reg}. Then Lemma~\ref{lem:proj-dyn} in Appendix~\ref{apx:proj-dyn} implies that there exists a solution $ (\hat\theta_1, r_1) $ to \eqref{eq:dyn-1} on $ \R_+ $ with $ (\hat\theta_1(0), r_1(0)) = (\hat\theta_0, r_0) $. Consider the corresponding observation error $ e_{\obs, 1} $ and switching signal $ \lambda_{e, 1} $ defined by \eqref{eq:obs-err} and \eqref{eq:est-sw}, and let
\begin{equation*}
    t_1 := \inf\{t > 0: \|e_{\obs, 1}(t)\| \leq \varepsilon_\obs'\}.
\end{equation*}
Then $ (\hat\theta_1, r_1) $ is a solution to \eqref{eq:est-dyn} and \eqref{eq:ctrl-dyn} on $ [0, t_1) $ with $ (\hat\theta_1(0), r_1(0)) = (\hat\theta_0, r_0) $. If $ t_1 = \infty $ then the proof is complete. Otherwise, $ e_{\obs, 1}(t_1) = \varepsilon_\obs' $ and thus $ \lambda_{e, 1}(t_1) = 0 $, and we continue with the second step below.

Second, consider the system defined by \eqref{eq:est-dyn} and \eqref{eq:ctrl-dyn} with $ \lambda_e \equiv 0 $, that is,
\begin{equation}\label{eq:dyn-2}
\begin{aligned}
    \dot{\hat\theta} &= 0, \\
    \dot r &= \big[ {-\lambda_r} \jacob{r} \hat J(r, \hat\theta)^\top \big]_{T_\cR(r)},
\end{aligned}
\end{equation}
which can also be modeled using the projected dynamical system \eqref{eq:proj-dyn} in Appendix~\ref{apx:proj-dyn} with the state $ x :=
(\hat\theta, r) $ and the set $ \cS := \Theta \times \cR $. Then Lemma~\ref{lem:proj-dyn} in Appendix~\ref{apx:proj-dyn} implies that there exists a solution $ (\hat\theta_2, r_2) $ to \eqref{eq:dyn-2} on $ [t_1, \infty) $ with $ (\hat\theta_2(t_1), r_2(t_1)) = (\hat\theta_1(t_1), r_1(t_1)) $. Consider the corresponding observation error $ e_{\obs, 2} $ and switching signal $ \lambda_{e, 2} $ defined by \eqref{eq:obs-err} and \eqref{eq:est-sw}, and let
\begin{equation*}
    t_2 := \inf\{t \geq t_1: \|e_{\obs, 2}(t)\| \geq \varepsilon_\obs\}.
\end{equation*}
Then $ (\hat\theta_2, r_2) $ is a solution to \eqref{eq:est-dyn} and \eqref{eq:ctrl-dyn} on $ [t_1, t_2) $ with $ (\hat\theta_2(t_1), r_2(t_1)) = (\hat\theta_1(t_1), r_1(t_1)) $. If $ t_2 = \infty $ then the proof is complete. Otherwise, $ e_{\obs, 2}(t_2) = \varepsilon_\obs $ and thus $ \lambda_{e, 2}(t_2) = \lambda_\theta $, and we continue with the first step above.

In this way, we obtain an increasing sequence $ (t_k)_{k \in \N} $ with $ t_0 = 0 $ and a corresponding sequence $ (\hat\theta_k, r_k)_{k \geq 1} $ of absolutely continuous functions $ \hat\theta_k: [t_{k-1}, \infty) \to \Theta $ and $ r_k: [t_{k-1}, \infty) \to \cR $. Moreover, following \eqref{eq:obs-err}, \eqref{eq:dyn-1}, \eqref{eq:dyn-2}, Assumption~\ref{ass:reg}, and the property that $ (\hat\theta_k(t), r_k(t)) \in \Theta \times \cR $ for all $ k \geq 1 $ and $ t \geq t_{k-1} $, there exists a constant $ M \geq 0 $ such that $ \|\dot e_{\obs, k}(t)\| \leq M $ for all $ k \geq 1 $ and $ t \geq t_{k-1} $. Hence
\begin{equation*}
    t_k - t_{k-1} \geq (\varepsilon_\obs - \varepsilon_\obs')/M \qquad \forall k \geq 2,
\end{equation*}
and thus $ \lim_{k \to \infty} t_k = \infty $. Therefore, the absolutely continuous functions $ \hat\theta: \R_+ \to \Theta $ and $ r: \R_+ \to \cR $ defined by
\begin{equation*}
    \hat\theta(t) := \hat\theta_k(t), \quad r(t) := r_k(t), \qquad k \geq 1, t \in [t_{k-1}, t_k)
\end{equation*}
form a solution to \eqref{eq:est-dyn} and \eqref{eq:ctrl-dyn} on $ \R_+ $ with $ (\hat\theta(0), r(0)) = (\hat\theta_0, r_0) $. The proof is completed by noticing that \eqref{eq:l-inf} follows from \eqref{eq:obs-err}, \eqref{eq:dyn-1}, \eqref{eq:dyn-2}, Assumption~\ref{ass:reg}, and the property that $ (\hat\theta(t), r(t)) \in \Theta \times \cR $ for all $ t \geq 0 $.

\bibliographystyle{IEEEtran}
\bibliography{reference-abbr}
\end{document}